\documentclass{article}

\usepackage[utf8]{inputenc} 
\usepackage[T1]{fontenc}    
\usepackage{xcolor}         
\usepackage{url}            
\usepackage{nicefrac}       
\usepackage{microtype}      
\usepackage{enumitem}
\usepackage{comment}

\usepackage{fullpage}
\usepackage{graphicx, wrapfig, subfigure} 
\usepackage{mathtools, amsthm, mathrsfs, amssymb, bbm}
\usepackage{thmtools}
\usepackage{comment}
\usepackage{natbib}
\usepackage[colorlinks=true,urlcolor=blue,citecolor=blue,linkcolor=blue,linktocpage,pdfpagelabels,bookmarksnumbered,bookmarksopen]{hyperref}
\usepackage[capitalize,noabbrev]{cleveref}
\usepackage{todonotes}
\usepackage{algorithm}
\usepackage[noend]{algpseudocode}
\usepackage{dsfont}
\usepackage{soul}
\usepackage{cancel}
\usepackage{forest}
\usepackage{tikz}

\usepackage{color-edits}
\addauthor{eb}{blue}
\addauthor{vg}{red}
\addauthor{ja}{green}
\addauthor{jc}{orange}

\newtheorem{definition}{Definition}

\newtheorem{lemma}{Lemma}
\newtheorem{claim}{Claim}

\newtheorem{condition}{Condition}

\crefname{condition}{Condition}{Conditions}
\Crefname{condition}{Condition}{Conditions}

\newcommand{\cost}[1]{\textup{cost}(#1)}
\newcommand{\charge}[1]{\textup{charge}(#1)}

\DeclareMathOperator{\alg}{\textsc{Alg}}

\DeclareMathOperator{\opt}{\textsc{Opt}}
\DeclareMathOperator{\osr}{\textsc{Osr}}

\DeclareMathOperator{\E}{\mathbb{E}}
\DeclarePairedDelimiter\ceil{\lceil}{\rceil}

\newcommand{\cA}{\mathcal{A}}
\newcommand{\cB}{\mathcal{B}}

\newcommand{\cD}{\mathcal{D}}

\newcommand{\leftBucket}{L}
\newcommand{\rightBucket}{R}

\newcommand{\learnBuckets}{t_0}

\DeclareMathOperator*{\argmax}{arg\,max}
\DeclareMathOperator*{\argmin}{arg\,min}

\title{Online Min-Cost Matching with General Arrivals\thanks{Eric Balkanski and Jason Chatzitheodorou were supported in part by NSF grants CCF-2210501 and IIS-2147361. Josh Ascher and Vasilis Gkatzelis were supported by NSF CAREER award CCF-2047907 and NSF grant CCF-2210502.}}

\author{%
    \begin{tabular}{c@{\hspace{1.5cm}}c}
        \begin{tabular}{c}
            Josh Ascher\\
            Drexel University\\
            \texttt{ja3443@drexel.edu}
        \end{tabular}
        &    
        \begin{tabular}{c}
            Eric Balkanski\\
            Columbia University\\
            \texttt{eb3224@columbia.edu}
        \end{tabular}
        \\
        \\
        \begin{tabular}{c}
            Jason Chatzitheodorou\\
            Columbia University\\
            \texttt{ic2621@columbia.edu}
        \end{tabular}
        &
        \begin{tabular}{c}
            Vasilis Gkatzelis\\
            Drexel University\\
            \texttt{gkatz@drexel.edu}
        \end{tabular}     
    \end{tabular}
}
\date{}

\begin{document}

\maketitle

\begin{abstract}

In the classic online min-cost matching problem, the goal is to match a sequence of requests that arrive dynamically over time to a set of static servers, aiming to minimize the total cost of the matching. This assumes that there are two distinct ``sides'' and that only one of these sides arrives online, but many of the motivating applications violate these assumptions. We study online min-cost perfect-matching when \emph{all} participants arrive online and, upon arrival, they need to either be matched to someone from a waiting pool or to join the waiting pool. We evaluate the competitive ratios achievable in different input models and show that for both the adversarial and the random-order input models the competitive ratio of any algorithm is unbounded. In contrast, for i.i.d.\ arrivals we give a $O( \log^2{n})$-competitive algorithm, even if the distribution that generates these arrivals is unknown to the algorithm. This result implies a rare example of separation in the achievable competitive ratio between the random-order and the unknown-i.i.d.\ input models.


    

\end{abstract}

\newpage

\section{Introduction}

In the metric matching problem, a set of points that lie in a metric space need to be matched to each other, aiming to minimize the total distance of the matching. 
Most prior work on this problem has focused on \emph{online} variants, where the points arrive dynamically over time and irrevocable matching decisions need to be made upon their arrival. The online metric matching problem was introduced and studied more than three decades ago \citep{KalyanasundaramPruhs1993,KMV94}, but it has recently received renewed interest, motivated by a surge of new applications in modern platform economies, such as ride-hailing or ride-sharing (e.g., Uber and Lyft), food delivery (e.g., Grubhub and DoorDash), player matching in online games (e.g., FlexMatch and SmartMatch), labor markets (e.g., Upwork and TaskRabbit), and kidney exchange. 

Motivated by some of the earlier applications, e.g., job requests being assigned to servers, the traditional model for online metric matching focuses on ``bipartite'' settings and ``one-sided arrivals'': there are two types of points (servers and requests), and while the requests arrive over time, the servers are all available at the beginning. However, this one-sided arrival model does not capture many of the modern applications, such as ride-sharing, where both the drivers and the passengers may join over time. In fact, many applications, such as real-time matchmaking of players in online games (e.g., matching chess players online based on their Elo rating), are not even bipartite, as any two participants in the market can potentially be matched to each other, depending on the compatibility of the match (e.g., the Elo rating distance). Although non-bipartite models have been well-studied in the value maximization variants of online matching,
they remain poorly understood for min-cost matching, with only a few exceptions (see Section~\ref{sec:related_work} for a detailed discussion).

The vast majority of work on online matching has been restricted to worst-case analysis, assuming that both the locations of the points in the metric and the order in which they arrive are adversarially chosen. This often leads to overly pessimistic bounds, so recent work has instead considered settings where the arrival order is chosen uniformly at random (the ``secretary'' setting) or the points are drawn i.i.d.\ from some known distribution. 
Another natural model 
is ``unknown-i.i.d.'', which drops the, often unrealistic, assumption that the distribution is known. These four input models were at the center of the classic survey by \citet[Theorem 2.1]{SurveyMehta13}, who observed that they 
give rise to the following hierarchy in terms of the best achievable competitive ratios (CR) for each one:

$$\text{CR(Adversarial)} \geq \text{CR(Random Order)} \geq \text{CR(Unknown-i.i.d.)} \geq \text{CR(Known-i.i.d.)}. 
$$

Developing a deeper understanding of this hierarchy in metric matching and identifying which of these inequalities are strict was also one of the open questions in \cite{gupta2019stochasticonlinemetricmatching}.

Our main goal in this paper is to move beyond one-sided arrival models in metric matching and to evaluate the competitive ratios achievable by online algorithms in this hierarchy above.
We focus on the well-studied line metric, which is acknowledged as an important metric space for online min-cost matching~\citep{peserico2021matching,antoniadis2019left,koutsoupias2003online,fuchs2005online}, and also closely captures some of the motivating applications (e.g., summarizing the skill level of chess players using their Elo ratings).

\subsection{Our Results} 

We investigate online min-cost matching on the line, beyond the classic model of one-sided arrivals, and we analyze both bipartite settings, which have two types of requests (e.g., drivers matched to passengers), and non-bipartite ones (e.g., game players matched to each other). Crucially, \emph{all} requests arrive sequentially, and upon the arrival of a request $r$ we must irrevocably decide whether to match it with an available request $r'$ from a waiting pool, or to reject all available matches and add it to the pool. Requests placed in the waiting pool need to be matched with a future request. A match incurs a cost equal to the distance $|r-r'|$, and our objective is to design algorithms that produce a perfect matching while minimizing the aggregate cost.

This \emph{general arrivals} model has received significant attention in the online max-value matching literature~(e.g., \citep{wang_online_2013, gamlath2019, tang2022,buchbinder2019online}), but not in min-cost matching (see Section~\ref{sec:related_work} for context). Our main result bridges this gap by proposing a novel online algorithm for min-cost matching in this model, which achieves a polylogarithmic competitive ratio for the non-bipartite case and the unknown-i.i.d.\ model.

\begin{restatable}{theorem}{UnknownIIDRatio}
    \label{thm: unknown iid ratio}
        There exists an $O(\log^2 n)$-competitive algorithm for non-bipartite online min-cost matching on the line with general arrivals in the unknown-i.i.d.\ model. 
\end{restatable}

Our algorithm starts by observing a fraction of the arriving requests without matching them. This way, it gathers a sample from the unknown distribution before making any, potentially costly, matching decisions, and it uses this sample to partition the line into intervals with a roughly equal number of requests. The rest of the algorithm proceeds in a sequence of phases, during which any new request that arrives in a non-empty interval (i.e., an interval with an unmatched request) is matched to a request in that interval. At the end of each phase, some of the intervals are merged into bigger intervals until all requests are matched.

If at any point during this process, the number of unmatched requests equals the number of requests yet to come, the algorithm is forced to match every  subsequent request, to ensure that it reaches a perfect matching. This introduces an interesting trade-off: choosing not to match some requests and place them in a waiting pool provides the algorithm with more context and with more matching options for each subsequent request. However, as the size of the waiting pool increases, there is a higher risk that the algorithm may have to make forced matches that could be very costly. For example, if we were to observe the first $n/2$ requests without matching them, we would then need to match each subsequent request to one of the $n/2$ initial ones. This essentially reduces the problem to the bipartite setting with one-sided arrivals, but no matter how we match the subsequent requests it is impossible to achieve a competitive ratio better than $O(\sqrt{n})$ (see Section~\ref{sec:framework} for details). 

Our algorithm succeeds by initially maintaining smaller intervals, which sets a high bar regarding what constitutes a good match, and avoids high matching costs due to a non-representative sample. Then, as time goes by, the interval size is increased and the matching becomes more permissive. It is worth noting that our algorithm is \emph{ordinal}, i.e., it does not use the actual location of a request on the line and instead only requires its position relative to other requests (i.e., the order of the requests on the line). As a result, its decision regarding whether to match two requests $r$ and $r'$ is made without knowing their actual distance, $d(r,r')$.

\paragraph{Random order model (Section~\ref{sec:random_order})}
Apart from the unknown-i.i.d.\ model, we also consider the random order model and we prove the following strong impossibility result:
    \begin{restatable}{theorem}{RandomOrderImpossibility}
    \label{thm: impossibility result in random order}
        The competitive ratio of any algorithm (even randomized ones) for online min-cost matching on the line with general arrivals in the random order model is unbounded for any $n\geq 4$.
    \end{restatable}

What is interesting about this result is that, combined with our main result, it exhibits a big separation between the competitive ratios achievable in the random order model and in the unknown-i.i.d.\ one. In his survey of online matching, \citet{SurveyMehta13} emphasized that \emph{``despite their apparent differences there is no known separation result differentiating the unknown-i.i.d.\ model from the random order model.''} To the best of our knowledge, this is the first example of separation between these two models in online matching.

\paragraph{Bipartite case (Section~\ref{sec:bipartite})}
Finally, for the bipartite case we provide a reduction that allows us to adapt results from the one-sided arrivals model to the general arrivals model. This reduction implies a constant competitive ratio for the line in the known-i.i.d.\ model, and an $O\left( (\log{\log{\log{n}}})^{2} \right)$ competitive ratio for general metrics.
    \begin{restatable}{theorem}{TwoSidedRatio}
    \label{thm: ratio for two sided bipartite arrivals}
     There exists a constant-competitive algorithm for bipartite online min-cost matching on the line with general arrivals in the known-i.i.d.\ model. For general metrics, there exists an $O\left( (\log{\log{\log{n}}})^{2} \right)$-competitive algorithm.
    \end{restatable}

\subsection{Related Work}\label{sec:related_work}

\paragraph{Min-cost matching with one-sided arrivals.} Online min-cost matching was first studied by~\citet{KalyanasundaramPruhs1993} and \citet{KMV94} who considered the one-sided arrivals model and gave a $2n-1$ competitive algorithm for general metrics. In the one-sided arrivals model, there is an equal number of servers and requests that lie in a metric space, the servers are known offline, and the online requests must be matched to an available server upon arrival. This model has since then been extensively studied and the line metric has received particular attention (see, e.g., \citep{KalyanasundaramPruhs1993, KMV94, fuchs2005online, koutsoupias2003online, raghvendra2018optimal, gupta2012online, nayyar2017input, peserico2021matching, balkanski2023power, megow2020online, antoniadis2019left}).  
The first  polylogarithmic competitive ratio was obtained by \citet{meyerson2006randomized}. The best-known upper and lower bounds are $O(\log^2 n)$ and $\Omega(\log n)$ for general metrics~\cite{bansal2007log2} and $O(\log n)$~\cite{gupta2012online} and $\Omega(\sqrt{\log n})$~\cite{peserico2021matching} for the line; closing these gaps remains a notable open problem in online algorithms.

\paragraph{Max-value matching with general arrivals.} In general arrival models, instead of having servers that are known offline, the points all arrive online. The general arrivals model considered in this paper is well-studied in online max-value matching (see, e.g., \citep{wang_online_2013, gamlath2019, tang2022,buchbinder2019online}), where it is referred to as the general vertex arrivals model. To the best of our knowledge, this model has not previously been studied in min-cost matching.

\paragraph{Min-cost  matching with general arrivals.} General arrival models that have been studied for min-cost matching include matching with delays, where the objective is to minimize the sum of the distances between matched pairs plus the sum of the delays, which are the waiting time between when a request arrives and when it is matched~\citep{emek2016onlinematchinghastemakes,bienkowski2018primaldualonlinedeterministicalgorithm,azar2020DeterministicMinCostMatchingDelays}. \cite{mari2025OnlineMatchingWithDelaysStochasticArrivalTimes} study a stochastic version of this model where arrivals follow a Poisson distribution. Instead of delays, \citet{kanoria2025dynamic} considers a bipartite model with supply and demand units  where, upon each arrival of a demand unit, a supply unit also arrives.

\paragraph{Non-bipartite matching.} In bipartite settings, the points are partitioned into two sides and a point on one side can only be matched to a point on the other side. In online non-bipartite matching, there is no such restriction, and every pair of points is a feasible match. The one-sided arrival model is only well-defined for the bipartite setting. In general arrival models, the primary focus has been on the non-bipartite (i.e., unrestricted) setting, but the bipartite setting has also been studied.

\paragraph{Online matching in the unknown-i.i.d.\ model.} While online min-cost matching has been mostly studied in the adversarial model, there has also been some work on the random order model~\citep{ raghvendra2016robust, Caragiannis_SerialDictatorship} and the known-i.i.d. model~\citep{gupta2019stochasticonlinemetricmatching, balkanski2023power, yang2026online, kanoria2025dynamic, tsai1994average}. For  min-cost matching with one-sided arrivals, \citet{raghvendra2016robust} gave an $O(\log n)$-competitive algorithm for general metrics in the random arrival order model and, for the known-i.i.d. model, \citet{gupta2019stochasticonlinemetricmatching} gave an algorithm that is $O((\log \log \log n)^2)$ and $O(1)$ competitive for general metrics and the line metric respectively.
The unknown-i.i.d. model has been discussed as a relevant intermediary model between the random-order and known-i.i.d. models~\citep{SurveyMehta13, gupta2019stochasticonlinemetricmatching, huang2024online, karande2011online}. To the best of our knowledge, no previous separation was known between the optimal competitive ratios for the random order and  unknown-i.i.d. models in online matching; and the best-known competitive ratios for the unknown-i.i.d. model were the ones inherited from results for the random order model. A related model to unknown-i.i.d. was recently considered in \cite{li2026smoothedanalysisonlinemetric}, where there are unknown distributions, and the algorithm is provided with a single sample from each distribution.

\paragraph{Matching using Ordinal Information}
The fact that our algorithm requires only ordinal information, i.e., the order of the requests on the line rather than their actual locations, is aligned with the literature on metric matching distortion. This line of work aims to design minimum cost metric matching algorithms that take as input only ordinal information (each request ranks all potential matches from closest to farthest), and it has analyzed ordinal algorithms for both general metrics~\citep{Caragiannis_SerialDictatorship,distortion_metric} and for the line~\citep{distortion_line}.

\section{Preliminaries}

 \paragraph{The online min-cost matching with general arrivals problem.} An instance $I$ consists of $n$ requests  embedded in a metric space  $(\mathcal{M},d)$, where $n$ is assumed to be even. The requests arrive online in some order, which we denote by a permutation $(r_1, \ldots, r_n)$ of $I$. At any time $t$, each request is in exactly one of three states: matched, unmatched, or not yet arrived.  Upon arrival of  request $r_t$ at time $t$, the algorithm must irrevocably match $r_t$ to an unmatched request, or leave it unmatched to be matched with a future request. At time $n$, all requests must be matched and the total cost of a matching $M$ is the sum of distances between matched requests, i.e., $\text{cost}(M) = \sum_{(r, r') \in M} d(r, r')$.

 \paragraph{The non-bipartite and bipartite settings.} In the bipartite setting, the requests $I$ are partitioned into equal-size sides,  and every request $r$ can only be matched to an unmatched request $r'$ that is on the other side.
We note that the  online  min-cost bipartite matching problem with one-sided arrivals is a special case of the bipartite setting of the online min-cost  matching with general arrivals problem where  the requests on one side  all arrive before the requests on the other side. In the non-bipartite setting, there is no such restriction and every match is feasible.

 \paragraph{The arrival order.} There are four standard arrival order models in online matching.
\begin{itemize}
    \item Adversarial: the arrival order is worst-case.
    \item Random: the arrival order is a uniformly random permutation of $I$.
    \item Unknown-i.i.d.: each request $r_t$ is drawn independently from an unknown distribution  $\mathcal{D}$.
    \item Known-i.i.d.: each request $r_t$ is drawn independently from a known distribution  $\mathcal{D}$.
\end{itemize}

 \paragraph{The competitive ratio.} 
 Let $\pi(I)$ be a uniformly random permutation of instance $I$ and let $\textsc{Alg}(\pi(I))$ denote the  total cost of the matching produced by an algorithm $\textsc{Alg}$ on the arrival sequence $\pi(I)$. Let $\textsc{Opt}(I)$ be the total cost of the optimal offline matching on $I$.  An algorithm $\textsc{Alg}$ is $c$-competitive in the random arrival order model if
\[
\max_{I} \ \frac{\mathbb{E}_{\pi,\textsc{Alg}}\!\left[\textsc{Alg}(\pi(I))\right]}{\textsc{Opt}(I)} \ \le \ c.
\]
 Let $\mathcal{D}$ be a distribution over $\mathcal{M}$ and $S=(r_1,\ldots,r_n)$ be a sequence of $n$ requests drawn i.i.d.\ from $\mathcal{D}$. In the unknown-i.i.d. and known-i.i.d. models, an algorithm  $\textsc{Alg}$ is $c$-competitive if
\[
\sup_{\mathcal{D}} \ \frac{\mathbb{E}_{S\sim \mathcal{D}^n,\textsc{Alg}}\!\left[\textsc{Alg}(S)\right]}{\mathbb{E}_{S\sim \mathcal{D}^n}\!\left[\textsc{Opt}(S)\right]} \ \le \ c.
\]

In the remainder of the paper, we write $\mathbb{E}\!\left[\textsc{Alg}(\pi(I))\right]$, $\mathbb{E}\left[\textsc{Alg}(S)\right]$, and $\mathbb{E}\left[\textsc{Opt}(S)\right]$ when the sources of randomness are clear from context.

\section{The Algorithmic Framework}\label{sec:framework}

In this section, we develop an algorithmic framework for general arrivals.  In Section~\ref{sec:3.1}, we describe the framework, which depends on a subroutine $\textsc{FindInitialBuckets}$ to be later defined. In Section~\ref{sec: framework main lemma}, we give the main lemma for this section, which shows that the competitive ratio is $O(\log^2 n)$ if $\textsc{FindInitialBuckets}$ satisfies three conditions. The proof of this lemma requires upper bounding the cost of the algorithm, which we do in Section~\ref{sec: upper bound cost of framework}, and lower bounding the optimal cost, which we do in Section~\ref{sec: opt lower bound}. As a warm-up, we then demonstrate how this framework allows us to prove an $O(\log^2 n)$ competitive ratio for the uniform distribution in Section~\ref{sec:warm-up}.

\paragraph{A reduction to the one-sided arrivals setting.} We first briefly discuss a simple reduction that uses an algorithm $\textsc{OneSided}$ for the one-sided arrivals setting. In this reduction, the first half  $R_1$ of the requests to arrive are left unmatched upon arrival. For the second half of the requests $R_2$, they are matched to an unmatched request from $R_1$ according to $\textsc{OneSided}$ that treats $R_1$ as the offline servers and $R_2$ as the online requests in the one-sided arrivals setting. 

For the uniform distribution, the expected optimal cost for general arrivals is constant (which is implied by  Lemma~\ref{lemma: opt pays expected rightmost minus expected leftmost}) and it is $\Theta(\sqrt{n})$ for one-sided arrivals ~\citep{akbarpour2021value,harel1993}. Thus, combined with the algorithm of~\citet{gupta2019stochasticonlinemetricmatching}  that is constant-competitive on the line for the known-i.i.d. with one-sided arrivals setting, this reduction gives an $O(\sqrt{n})$-competitive algorithm. However, even for the uniform distribution, there is a $\Theta(\sqrt{n})$ gap between the optimal costs in the general and one-sided arrival settings. Consequently, this reduction—and, more generally, any algorithm that leaves the first half of the requests unmatched upon arrival—cannot achieve a competitive ratio better than 
$\Theta(\sqrt{n})$.

\subsection{Description of the Algorithmic Framework}
\label{sec:3.1}

To improve over the $\Theta(\sqrt{n})$ competitive ratio achieved by the reduction to the one-sided arrival setting, an algorithm must thus start matching requests before time $t= n/2$. A tradeoff then arises since the earlier a request is matched, the fewer options there are to match that request. Our framework balances this tradeoff by matching an early request upon arrival only if a great match is available, and then over time, it  lowers the threshold for how good a match needs to be to occur. The algorithm, called \textsc{Framework} and described in Algorithm~\ref{alg:framework}, does leave the first $\learnBuckets$ requests unmatched, but for some parameter $\learnBuckets < n/2$. These $\learnBuckets$ requests are partitioned into buckets $B_i$ according to a subroutine \textsc{FindInitialBuckets}, which also constructs an initial mapping $f^1$ from locations $x \in [0,1]$ to a bucket index $f^1(x)$. This mapping plays a central role. 

The core of the algorithm is the following, which occurs over phases. For each request $r_t$ during a phase $k$, $r_t$ is mapped by $f^k$ to  bucket $B_i$ with  $i = f^k(r_t)$. If bucket $B_i$ contains some unmatched request, then $r_t$ is matched to an unmatched request in $B_i$. Otherwise, it is added to $B_i$. At the end of a phase, the number of buckets is reduced such that each new bucket is obtained by merging $\beta$ old buckets, for some parameter $\beta$. The mapping $f^k$ is also updated so that any request that would have been mapped to an old bucket $f^k(r)$ is now mapped to a new bucket $f^k(r) / \beta$ that the old bucket $f^k(r)$ got merged into. We highlight that, informally, reducing the number of buckets means that the algorithm becomes less restrictive for matching a newly arrived request.

\begin{algorithm}[h]
\begin{algorithmic}
    \caption{$\textsc{Framework}$}
    \label{alg:framework}
    \State \textbf{Input:}  number  of requests $n$,  learning parameter $0 \leq  \learnBuckets \leq n$
    \State merging parameter $\beta \leftarrow  \left\lfloor400 \log n\right\rfloor$
    \State Let number of phases $\rho$ be such that $ \frac{n}{4\beta^{2}} \leq \beta^{\rho} \leq \frac{n}{4\beta} $
    \State Initialize matching $M \gets \emptyset$ 
    \State  First $\learnBuckets$ requests  $r_1, \dots, r_{\learnBuckets}$ are left unmatched upon arrival
    \State $B_{\leftBucket}, B_1, \ldots, B_{\beta^{\rho}}, B_{\rightBucket}, f^1 \leftarrow \textsc{FindInitialBuckets}\left(\left\{r_1, \dots, r_{\learnBuckets}\right\}, \beta^{\rho}\right)$
    \For{phase $k= 1$ to $\rho$} 
    \State $t_k \leftarrow n - 3 \beta^{\rho + 1 - k}$
    \For{time steps $t = t_{k-1}$ to $t_k$} 
    \State Request $r_t$ arrives
    \State $i \leftarrow f^k(r_t)$ \Comment{find bucket index for current request}
     \If{$B_i  \neq \emptyset$} \Comment{match current request}
    \State $r \gets $ arbitrary request in $B_i$
         \State Add  $\{r_t, r\}$ to $M$ and remove $r$ from $B_i$ 
        \Else \Comment{do not match current request}
        \State Add $r_t$ to $B_i$
        \EndIf
        \If{\# of requests not yet arrived = \# of unmatched requests}
        \State \textbf{break} out of outer for loop
        \EndIf
    \EndFor
    \State $B_1,  \ldots, B_{\beta^{\rho - k}} \gets \cup_{i=1}^{\beta} B_i,  \ldots, \cup_{i=\beta^{\rho + 1 - k} - \beta+1}^{\beta^{\rho + 1 - k}} B_i$ \Comment{ merge buckets }
    \State $f^{k+1}(x) \gets 
        \begin{cases}
             f^k(x) & \text{ if  } f^{k}(x)  \in \{\leftBucket, \rightBucket\} \\
             \ceil*{\frac{f^{k}(x)}{\beta} } & \text{ otherwise }
        \end{cases}$ 
    \Comment{update bucket mapping}
    \EndFor

    \State Match each remaining request not yet arrived to an arbitrary unmatched request and update $M$

    \State \Return $M$
\end{algorithmic}
\end{algorithm}

There are two additional important features of the algorithm. The first is that there are two special buckets, the left and right buckets $B_{\leftBucket}$ and $B_{\rightBucket}$, that are for requests close to $0$ and $1$ respectively and that never get merged with other buckets. The second is that if the number of requests not yet arrived reaches the number of unmatched requests, then the algorithm breaks out of the phases and matches each remaining request to an arbitrary unmatched request.

\subsection{Main Lemma for \textsc{Framework}}
\label{sec: framework main lemma}

The main lemma of this section shows that if the subroutine \textsc{FindInitialBuckets} satisfies three conditions, then Algorithm~\ref{alg:framework} is $O(\log^2 n)$-competitive. The first condition is that, with high probability the break condition of the algorithm is not triggered. We note that at time $t$, the  unmatched requests are $\{r_1, \dots, r_{t}\} \setminus M_t$, where $M_t$ is the set of matched requests at time  $t$, and the number of requests not yet arrived is $n - t - 1$. We also note that the break condition cannot be triggered after time $t_\rho$, which is when the final iteration of the outer loop finishes.

\begin{condition}
\label{cond: probability of hitting wall}
     $\Pr[\forall t \in \left[t_{\rho}\right] : |\{r_1, \dots, r_{t}\} \setminus M_t | < n - t - 1 ] \geq 1-\frac{\rho}{n}$.
\end{condition}

We let $i(t)$ be the index $i$ of the bucket $B_i$ to which request $r_t$ is mapped at time step $t$ of the algorithm. The second condition is that, for any phase and bucket, a sufficiently small number of requests are assigned to that bucket during that phase. For the left and right buckets $B_{\leftBucket}$ and $B_{\rightBucket}$, the total number of requests assigned to these buckets must be bounded.

\begin{condition}
\label{cond: few arrivals in any one bucket} 
    We have the following bounds on the number of requests in each bucket:
    \begin{enumerate}
        \item For buckets $B_{\leftBucket}$ and $B_{\rightBucket}$, we have 
        \[
        \Pr[\left| \{t \in \{ \learnBuckets+1, \dots, n\}: i(t) \in \left\{\leftBucket, \rightBucket\right\} \} \right| \leq 32\beta^2] \geq 1-1/n.
        \]
        \item For all other buckets in the first phase, we have 
        \[
        \Pr\left[ \forall j \in \left[\beta^{\rho}\right]: \left| \{t \in \{\learnBuckets+1, \dots, t_{1}\}: i(t) = j \} \right| \leq 32\beta^2 \right] \geq 1-1/n.
        \]
        \item For all other buckets in subsequent phases, we have 
        \[
        \Pr[ \forall k \in [2..\rho], j \in \left[\beta^{\rho + 1 - k}\right]: \left| \{t \in \{ t_{k-1}+1, \dots, t_{k}\}: i(t) = j \} \right| \leq 24\beta ] \geq 1-1/n.
        \]
    \end{enumerate}
\end{condition}

We let $I^k_j = \{ x \in [0,1]: \Pr\left[ f^k(x) = j \right] > 0\}$ be the locations that have a non-zero probability of being assigned to bucket $B_j$ during phase $k$. The third and last condition states that the initial mapping $f^1$ from locations to buckets is such that the sum, over the $\beta^\rho$ buckets $B_j$ in the first phase, of the lengths of the intervals $I^1_j$ is, up to a factor $3$, upper bounded by the expected distance between the leftmost and rightmost requests.

\begin{condition}\label{cond: buckets have bounded length}
    $\E_{I \sim \cD^n}\left[\sum_{j \in \left[\beta^{\rho}\right]} (\sup{I^1_j} - \inf{I^1_j})\right] 
    \leq 3 \cdot \E_{I \sim \cD^n}\left[ \max_{t \in [n]}r_{t} - \min_{t \in [n]}r_{t} \right]$.
\end{condition}

The following lemma is the main result for the framework. Any subroutine \textsc{FindInitialBuckets} that satisfies \cref{cond: probability of hitting wall,cond: few arrivals in any one bucket,cond: buckets have bounded length}, achieves competitive ratio $O\left(\log^2 n\right)$. 
The proof follows from upper bounding the expected cost of the framework (\cref{lem: upper bound cost of framework}) in \cref{sec: upper bound cost of framework} and lower bounding the expected cost of the offline optimal solution (\cref{lemma: opt pays expected rightmost minus expected leftmost}) in \cref{sec: opt lower bound}.

\begin{restatable}{lemma}{MainLemma}
\label{lem: ratio of algorithmic framework}
    For any  distribution $\mathcal{D}$, if \textsc{FindInitialBuckets} and the learning parameter $\learnBuckets$ and such that \cref{cond: probability of hitting wall}, \cref{cond: few arrivals in any one bucket}, and \cref{cond: buckets have bounded length} are satisfied, then \cref{alg:framework} is $O(\log^2{n})$-competitive.
\end{restatable}

\subsection{Upper Bound for Cost of \textsc{Framework}}
\label{sec: upper bound cost of framework}

In this section, we upper bound  the expected cost of \textsc{Framework}. The cost of \textsc{Framework} over
$I \sim \cD^n$ is denoted as $\alg(I)$.

\begin{restatable}{lemma}{UpperBoundFramework}
\label{lem: upper bound cost of framework}
    For any  distribution $\mathcal{D}$, if \textsc{FindInitialBuckets} and the learning parameter $\learnBuckets$ are such that \cref{cond: probability of hitting wall}, \cref{cond: few arrivals in any one bucket}, and \cref{cond: buckets have bounded length} are satisfied,
    then
    $$\E_{I \sim \cD^n}\left[ \alg(I) \right] 
    \leq O\left( \log^2{n} \right) \cdot \E_{I \sim \cD^n}\left[ \max_{t \in [n]}r_t - \min_{t \in [n]}r_t \right].$$
\end{restatable}

The remainder of \cref{sec: upper bound cost of framework}  gives an overview of the proof of this lemma.
For ease of presentation, we let $\tau(t)$ denote the time step, such that $\{r_{t}, r_{\tau(t)}\} \in M$.
We start by defining the cost incurred by \cref{alg:framework} on each time step $t$:
\begin{align*}
    \cost{t} = 
        \begin{cases}
            |r_t - r_{\tau(t)}|,&\text{ if } t > \tau(t) \\
            0,                  &\text { else}
        \end{cases}
\end{align*}

Next, we introduce some notation.  Let the CDF of $\cD$ be $F(x) = \Pr_{r \sim \cD}[r\leq x]$, the upper quantile $F^{-1}_+(p) = \inf\{x \in [0,1] : F(x) \geq p\}$,
and the lower quantile $F^{-1}_-(p) = \sup\{x \in [0,1] : F(x) \leq p\}$. The proof relies crucially on the following interval $J$.
\begin{align*}
    J = [F^{-1}_-(1/n), F^{-1}_+(1-1/n)]
\end{align*}
We also define a second notion of cost that depends on whether request $r_t$ is in $J$ or not.
\begin{align*}
    \charge{t} = 
        \begin{cases}
            \sup{J} - r_t  &\text{ if } r_t < \inf{J} \\
            r_t - \inf{J}  &\text{ if } r_t > \sup{J} \\
            0             &\text{ if } r_t \in J, r_{\tau(t)} \notin J \\
            \frac{|r_t - r_{\tau(t)}|}{2} &\text{ if } r_t \in J, r_{\tau(t)} \in J 
        \end{cases}
\end{align*}
Intuitively, each request that arrives in a low probability area (outside of $J$) contributes a lot to the cost, as it will probably be matched to a request in a high probability area (inside $J$). Let 
$$A_{\rho} = \{ \forall t \in [t_\rho]: |\{r_1, \dots, r_{t}\} \setminus M_t | < n - t - 1 \},$$
$$E_{a} = \left\{ \forall k \in [\rho], j \in \left[\beta^{\rho + 1 - k}\right]: \left| \{t \in \{ t_{k-1}+1, \dots, t_{k}\}: i(t) = j \} \right| \leq 12\beta \right\},$$
$$E_o = \{ \left| \{t \in \{ \learnBuckets+1, \dots, n\}: i(t) \in \left\{\leftBucket, \rightBucket\right\} \} \right| \leq \beta^2 \}.$$
be the events that the first two conditions require to occur with high probability. The bad event $E_{bad}$ is the event that at least one of $A_{\rho}$, $E_{a}$, and  $E_o$ does not occur, i.e.,  $E_{bad} = \bar{A}_{\rho} \cup \bar{E}_{a} \cup \bar{E}_{o}$.

We split the proof of \cref{lem: upper bound cost of framework} into three separate cases.  The first case concerns the cost of all the requests that arrive within $J$, conditional on event $E_{bad}$ happening, which we charge to the length of $J$. The  proofs of \cref{lem:maxminregions,lem: charge requests in J for bad event,lem: total cost for good event,lem: charge requests outside J}, as well as the proof of \cref{lem: upper bound cost of framework} that combines these four  lemmas, are deferred to \cref{sec: upper bound cost of framework omitted proofs}.
\begin{restatable}{lemma}{ChargeInsideJBadEvent}
\label{lem: charge requests in J for bad event}
    For any distribution $\cD$, we have
    $$\E_{I \sim \cD^n}\left[ \sum_{t \in [n]: r_t \in J}\charge{t} \;\middle|\; E_{bad} \right] \leq n \cdot \left( \sup{J} - \inf{J} \right).$$
\end{restatable}

The second case concerns the cost incurred by our algorithm, conditional on $E_{bad}$ not happening, which we charge to the expected distance between the rightmost and leftmost request.
\begin{restatable}{lemma}{CostGoodEvent}
\label{lem: total cost for good event}
    For any  distribution $\mathcal{D}$, if \textsc{FindInitialBuckets} and learning parameters $\learnBuckets$ is such that \cref{cond: probability of hitting wall}, \cref{cond: few arrivals in any one bucket}, and \cref{cond: buckets have bounded length} are satisfied, then
    $$\E_{I \sim \cD^n}\left[ \sum_{t \in [n]}\cost{t} \;\middle|\; \bar{E}_{bad} \right] 
    \leq O(\log^2{n}) \cdot \E_{I \sim \cD^n}\left[ \max_{t \in [n]}r_t - \min_{t \in [n]}r_t \right].$$
\end{restatable}

The third case concerns the cost incurred by any algorithm for all requests that arrived outside of $J$, and is charged to the cost of matching them to a request inside $J$. Intuitively, even though this overestimates the cost, it is only done for arrivals far from $J$; therefore, the optimal solution pays their distance to $J$ as well.
\begin{restatable}{lemma}{ChargeOutsideJ}
\label{lem: charge requests outside J}
    For any distribution $\cD$, we have: 
    $$\E_{I \sim \cD^n}\left[ \sum_{t \in [n]: r_t \notin J}\charge{t} \right] \leq 2(\sup{J} - \inf{J}) + n \cdot \E_{r \sim \cD}\left[\inf_{x \in J} |r - x|\right].$$
\end{restatable}

We further show a technical lemma that allows us to relate the expected distance between the rightmost and the leftmost request to the cost of requests that arrive within the interval $J$ and those that arrive outside.
\begin{restatable}{lemma}{MaxMinRegions}
\label{lem:maxminregions}
    For any distribution $\cD$, we have:
    \[
    \mathbb E[\max_{t \in [n]} r_t-\min_{t \in [n]} r_t]
    \;\ge\;
    \frac14
    (\sup{J} - \inf{J})
    + \frac{n}{4} \E_{r \sim \cD}\left[\inf_{x \in J} |r - x|\right].
    \]
\end{restatable}

We are now ready to give a proof sketch of \cref{lem: upper bound cost of framework}, which combines the four previous lemmas.

\begin{proof}[Proof Sketch of \cref{lem: upper bound cost of framework} (full proof in \cref{sec: upper bound cost of framework omitted proofs})]
    First, note that the total cost incurred by the algorithm is upper bounded by the total amount we charge to each arrival, $\sum_{t \in [n]} \charge{t}$, so we focus on upper bounding this quantity in expectation.
    We split the total charge into three terms: (i) $\sum_{t \in J} \charge{t}$ conditional on event $E_{bad}$ happening, (ii) $\sum_{t \in J} \charge{t}$ conditional on event $E_{bad}$ not happening, and (iii) $\sum_{t \notin J} \charge{t}$.
    
    For (i) we use \cref{lem: charge requests in J for bad event} to obtain upper bound $n \cdot (\sup{J} - \inf{J})$, and $E_{bad}$ happens with probability $O(\log^2{n} / n)$, so we get $O(\log^2{n}) \cdot (\sup{J} - \inf{J}) \leq O(\log^2{n}) \cdot \E_{I \sim \cD^n}\left[ \max_{t \in [n]}r_t - \min_{t \in [n]}r_t \right]$ due to \cref{lem:maxminregions}.

    For (ii), notice that when $r_t \in J$, we have  $\charge{t} \leq \cost{t}$ by definition, so we upper bound this term by the total cost of the algorithm conditional on $E_{bad}$ not happening. 
    Then we use \cref{lem: total cost for good event} to obtain upper bound $O(\log^2{n}) \cdot \E_{I \sim \cD^n}\left[ \max_{t \in [n]}r_t - \min_{t \in [n]}r_t \right]$.

    For (iii) we use \cref{lem: charge requests outside J} to obtain upper bound $2 (\sup{J} - \inf{J}) + n \cdot \E_{r \sim \cD}\left[\inf_{x \in J} |r - x|\right]$, which we further upper bound by $O(1) \cdot \E_{I \sim \cD^n}\left[ \max_{t \in [n]}r_t - \min_{t \in [n]}r_t \right]$ by \cref{lem:maxminregions}.
\end{proof}

\subsection{Lower Bound for OPT}
\label{sec: opt lower bound}

In this section, we lower bound the expected cost of the optimal matching by, up to a factor of 2, the expected distance between the leftmost and rightmost requests.
\setcounter{equation}{0}
\begin{restatable}{lemma}{OPTMaxMin}
\label{lemma: opt pays expected rightmost minus expected leftmost}
    For any distribution $\cD$, we have
    $$\E_{I \sim \cD^n}\left[ \opt(I) \right] \geq \frac{1}{2} \cdot \E_{I \sim \cD^n}\left[ \max_{t \in [n]}{r_t} - \min_{t \in [n]}{r_t} \right].$$
\end{restatable}

To prove \cref{lemma: opt pays expected rightmost minus expected leftmost}, we show a useful property of the optimal solution, that greedily matching the requests from left to right achieves the minimum cost, if we know the locations of all $n$ requests. 
The proof is based on an exchange argument, since we can always exchange the left-most pair that is not matched according to the greedy matching without increasing the cost. An analogous version of this lemma is well known for the bipartite setting.

The full proofs of both lemmas are deferred to \cref{sec: opt lower bound omitted proofs}.
\begin{restatable}{lemma}{OPTStructure}
\label{lemma: OPT structure}
    Let $0 \leq r^{(1)} \leq r^{(2)} \leq \dots \leq r^{(n)} \leq 1$  and $I$ is any permutation of $r^{(1)}, \dots, r^{(n)}$.
    Then $\opt(I)= \sum_{i=1}^{n/2}\left( r^{(2i)} - r^{(2i - 1)} \right)$.
\end{restatable}

\subsection{Warm-Up: the Uniform Distribution}\label{sec:warm-up}

We  demonstrate the benefits of our framework for the special case of the (known) uniform distribution, as a warm-up.
The requests $r_1, \ldots, r_{\learnBuckets}$ are used for \textsc{FindInitialBuckets} to learn good buckets. When the distribution is known, there is no need to wait to observe requests, the mapping $f$ can be directly defined using the distribution. Thus, for our warm-up for the uniform distribution, we define $\textsc{FindInitialBuckets}$ assuming  it is provided with no requests and the initial buckets $B_{\leftBucket}, B_1, \ldots, B_b, B_{\rightBucket}$ are thus empty. For the mapping $f$,  the algorithm partitions the  $[0,1]$ interval into $b$ intervals and $f$ maps locations in the $i^{th}$ interval $[(i-1)/b, i/b]$ to the index $i$ for bucket $B_i$. Note that this also makes the algorithm non-ordinal, as we require the actual location of each request to determine which bucket to map them to. However, the main algorithm that we describe in the next section, which works even for unknown i.i.d.\ arrivals, defines the buckets and the mapping using only ordinal information.

\begin{algorithm}[H]
\begin{algorithmic}
    \caption{The $\textsc{FindInitialBuckets}(\emptyset, b)$ subroutine for the uniform distribution}
    \label{alg:warmup}
    \State \textbf{Input:} number of buckets $b$
    \State Initialize $B_{\leftBucket}, B_1, \ldots, B_{b}, B_{\rightBucket} \leftarrow \emptyset$
    \State Define $f(x) = \lceil x b \rceil$ for all $x \in [0,1]$
    \State \Return $B_{\leftBucket}, B_1, \ldots, B_{b}, B_{\rightBucket}, f$
\end{algorithmic}
\end{algorithm}

The main result for this $\textsc{FindInitialBuckets}$ is the following. 
\begin{restatable}{theorem}{WarmupRatio}
    Algorithm~\ref{alg:framework} with  $\learnBuckets = 0$ and the subroutine \textsc{FindInitialBuckets} defined in Algorithm~\ref{alg:warmup} achieves an $O(\log^2 n)$ competitive ratio for non-bipartite online min-cost matching on the line with general arrivals in the known-i.i.d. model with requests drawn from the uniform distribution.
\end{restatable}
\begin{proof}
    We proceed to show that \cref{alg:warmup} satisfies \cref{cond: probability of hitting wall,cond: few arrivals in any one bucket,cond: buckets have bounded length}, the $O(\log^2 n)$-competitive ratio then follows  immediately from \cref{lem: ratio of algorithmic framework}.

    \paragraph{\Cref{cond: probability of hitting wall}.}
    We prove this by inducting on the number of phases. First, note that if the number of unmatched requests is ever equal to the number of requests yet to arrive, it will be equal for all subsequent time steps as all arriving requests will be immediately matched and the number of unmatched requests and yet to arrive requests will decrease at the same rate. Therefore, it suffices to consider whether the number of unmatched requests is equal to the number of future requests at time $t_{\rho}$. 

    For $k \in [\rho]$, let $F_k = \cap_{j \in \left[\beta^{\rho+1-k}\right]} \left\{\left|B_j \setminus M_{t_k}\right| \leq 1\right\}$ be the event that each bucket ends the $k$-th phase with at most one unmatched request and $A_k = \left\{\forall t \in \left[t_{k}\right] : |\{r_1, \dots, r_{t}\} \setminus M_t | < n - t - 1 \right\}$ be the event that the number of unmatched requests at the end of the $k$-th phase is less than the number of future requests. We will show that, $\Pr[A_{\rho}] \geq \Pr\left[\cap_{k \leq \rho} F_k\right] \geq 1-\rho/n$. We first prove $\Pr\left[\cap_{k \leq \rho} F_k\right] \geq 1-\rho/n$ by induction. In the first phase, each bucket begins with zero unmatched requests. Since requests in the same bucket are matched immediately, it is impossible for some bucket to end the phase with more than one unmatched request, i.e. $\Pr[F_1] = 1 \geq 1-1/n$. 
    
    Now, fix some $k \geq 1$ and assume $\Pr\left[\cap_{k'\leq k} F_{k'}\right] \geq 1-k/n$. Note that 
    \begin{align*}
    \Pr\left[\cup_{k' \leq k+1} \bar{F}_{k'}\right] &=  \Pr\left[\cup_{k' \leq k} \bar{F}_{k'}\right] + \Pr\left[\bar{F}_{k+1}\middle|\cap_{k' \leq k} F_{k'}\right] \cdot \Pr\left[\cap_{k' \leq k} F_{k'}\right]\\
    &\leq \Pr\left[\cup_{k' \leq k} \bar{F}_{k'}\right] + \Pr\left[\bar{F}_{k+1}\middle|\cap_{k' \leq k} F_{k'}\right]\\
    &\leq \frac{k}{n} + \Pr\left[\bar{F}_{k+1}\middle|\cap_{k' \leq k} F_{k'}\right],
    \end{align*}
    where the first equality follows from the tower rule and the fact that $\cup_{k' \leq k} \bar{F}_{k'} \implies \cup_{k' \leq k+1} \bar{F}_{k'}$, and the last inequality follows from the induction hypothesis. So, it suffices to show $\Pr\left[\bar{F}_{k+1}\middle|\cap_{k' \leq k} F_{k'}\right] \leq \frac{1}{n}$. Note that, by $F_k$, each bucket ended the $k$-th phase with at most one unmatched request. Since the requests in the $(k+1)$-th phase are formed from $\beta$ buckets in the $k$-th phase, this means that each bucket begins the $(k+1)$-th phase with at most $\beta$ unmatched requests.
    
    Since each bucket in the $(k+1)$-th phase has probability $1/\beta^{\rho+1-(k+1)}$ of receiving a request,  the following Chernoff bound proves that $\Pr\left[\bar{F}_{k+1} \mid \cap_{k' \leq k} F_{k'}\right] \leq \frac{1}{n}$.    Let $X^{k+1}_j$ be the number of requests in bucket $B_j$. Then, by definition of $f$ in \Cref{alg:warmup}, we have:
    \begin{align*}
        \E_{I \sim \cD^n}\left[ X^{k+1}_j \right] = \E\left[\left| \{ t \in \{t_{k} + 1, \dots, t_{k+1}\}: i(t) = j\} \right|\right] &= \left( t_{k+1} - t_{k} \right) \cdot \frac{1}{\beta^{\rho + 1 - (k+1)}} = 3(\beta-1)
    \end{align*}

    Then, by Chernoff bounds, for $\delta = 1$:
    \begin{align*}
        \Pr_{I \sim \cD^n}\left[X_j^{k+1} \leq (1- \delta) \cdot \E_{I \sim \cD^n}\left[ X_j^{k+1} \right]  \right] 
        \leq \exp\left(- \frac{\delta^2}{2} (3\beta - 1)\right) = \exp\left(-\frac{3\beta-1}{2}\right)
        \leq \frac{1}{n},
    \end{align*}
    where the last inequality follows from $\beta = \left\lfloor 400\log{n} \right\rfloor$. So, we have
    \begin{align*}
    \Pr\left[\cup_{k' \leq k+1}\bar{F}_{k'}\right] \leq \Pr\left[\bar{F}_{k+1} \mid \cap_{k' \leq k} F_{k'}\right] + \frac{k}{n} \leq \frac{1}{n} + \frac{k}{n} =\frac{k+1}{n}.
    \end{align*}
    
    So, it suffices to show $\Pr[A_{\rho}] \geq \Pr\left[\cap_{k \leq \rho} F_k\right]$, or equivalently, $\cap_{k \leq \rho} F_k \implies A_{\rho}$. To do so, note that at the end of any phase, the number of future requests is equal to three times the number of buckets. By $\cap_{k \leq \rho} F_k$, each bucket ends each phase with at most one unmatched request. So, at the end of any phase, the number of unmatched requests is at most the number of buckets, which is less than the number of future requests. Thus, given $\cap_{k \leq \rho} F_k$, it is impossible that the number of unmatched requests at time $t_{\rho}$ is equal to the number of future requests. Therefore, we have $\Pr[A_{\rho}] \geq \Pr\left[\cap_{k \leq \rho} F_k\right] \geq 1-\rho/n$, as desired.

    \paragraph{\cref{cond: few arrivals in any one bucket}.}
    
    By definition of $f$, no requests are ever assigned to $B_{\leftBucket}$ or $B_{\rightBucket}$. Therefore, w.p.\ $1$, we have $\left| \{t \in \{ \learnBuckets+1, \dots, n\}: i(t) \in \left\{\leftBucket, \rightBucket\right\} \} \right| = 0 \leq \beta^2$.

    To prove the second part of the condition, again note that, for any $k \in [2..\rho], j \in \left[\beta^{\rho+1-k}\right]$, the expected number of arrivals in bucket $B_j$ is
    \begin{align*}
    \E\left[X_j^k\right] = \E\left[\left| \{ t \in \{t_{k-1} + 1, \dots, t_{k}\}: i(t) = j\} \right|\right]  &= \beta^{\rho+1-k} \cdot 3(\beta-1) \cdot \frac{1}{\beta^{\rho+1-k}}\\
    &= 3(\beta-1).
    \end{align*}
    Then, a Chernoff bound with $\delta = 3$ yields
    \[
    \Pr_{I \sim \cD^n} \left[X^k_j \geq (1+\delta) \cdot \E\left[X_j^k\right]\right] \leq \exp\left(-\frac{\delta^2 \cdot \E\left[X_j^k\right]}{2+\delta}\right) \leq \exp\left( - \frac{9 \cdot 3(\beta-1)}{5}\right) \leq \frac{1}{n},
    \]
    where the last inequality follows from $\beta = \left\lfloor 400\log{n} \right\rfloor$. For $k=1$, we have 
    \begin{align*}
    \E\left[X_j^1\right] = \E\left[\left| \{ t \in \{\learnBuckets + 1, \dots, t_{1}\}: i(t) = j\} \right|\right]  &= \left(n-\learnBuckets - 3\beta^{\rho}\right) \cdot \frac{1}{\beta^{\rho}}\\
    &< 4 \beta^{\rho+2} \cdot \frac{1}{\beta^{\rho}}\\
    &=4\beta^2.
    \end{align*}
    Then, a Chernoff bound with $\delta = 1$ yields
    \[
    \Pr_{I \sim \cD^n} \left[X^1_j \geq (1+\delta) \cdot \E\left[X_j^1\right]\right] \leq \exp\left(-\frac{\delta^2 \cdot \E\left[X_j^1\right]}{2+\delta}\right) \leq \exp\left( - \frac{4\beta^2}{3}\right) \leq \frac{1}{n},
    \]
    where the last inequality follows from $\beta = \left\lfloor 400\log{n} \right\rfloor$.
    
    \paragraph{\cref{cond: buckets have bounded length}.}
    We bound the expected length of the intervals:
    \begin{align*}
        \E_{I \sim \cD^n}\left[\sum_{j = 1}^{ \beta^{\rho} }  (\sup{I^1_j} - \inf{I^1_j}) \right]
        = \beta^{\rho} \cdot \frac{1}{\beta^{\rho}}
        = 1 
        \leq 3 - \frac{6}{n+1}
        = 3 \cdot \E_{I \sim \cD^n}\left[ \max_{t \in [n]} r_t - \min_{t \in [n]} r_t \right],
    \end{align*}
    where the first equality is since $I^1_j = \left[ \frac{(j-1)}{\beta^{\rho}}, \frac{j}{ \beta^{\rho}} \right]$ for all $j \in [ \beta^{\rho} ]$ by definition of $f$, the second inequality holds for $n > 0$ and the second equality since $\E_{I \sim \cD^n}\left[ \max_{t \in [n]}{r_t} - \min_{t \in [n]}{r_t} \right] = \frac{n}{n+1} - \frac{1}{n+1} = \frac{n-1}{n+1} = 1 - \frac{2}{n+1}$ by order statistics of $U[0,1]$.
\end{proof}

\section{Unknown-i.i.d.}\label{sec: unknown iid}

In this section, we describe the \textsc{FindInitialBuckets} subroutine for the unknown-i.i.d. setting (Section~\ref{sec:FIBdes}) and then show that it satisfies the three conditions needed for the framework to obtain an $O(\log^2 n)$ competitive ratio (Section~\ref{sec:FIBanalysis}). 

\subsection{Description of \textsc{FindInitialBuckets} for the Unknown-i.i.d. Setting}
\label{sec:FIBdes}

 \textsc{FindInitialBuckets} for the unknown-i.i.d. setting (Algorithm~\ref{alg:FIBordinal}) uses the set of requests $T$, which are not matched upon arrival by the framework,  as training data to learn good initial buckets and mapping $f$. A first challenge is that if the expected optimal cost is significantly smaller than the optimal cost for the uniform distribution, which occurs when the distribution $\mathcal{D}$ concentrates on a small region of the $[0,1]$ interval, then the algorithm needs to be more demanding about when to match requests.  
 
 To accomplish that, the mapping $f$ assigns to the special buckets $B_{\leftBucket}$ and $B_{\rightBucket}$, which do not get merged with other buckets at the end of each phase of the framework,    the requests to the left of $\min(T)$ or to the right of $\max(T)$. This assignment implies that the algorithm has reduced the matching problem from the $[0,1]$ interval to the $\left[\min(T), \max(T)\right]$ interval. In particular, when $\mathcal{D}$ concentrates on a small region, then $\left[\min(T), \max(T)\right]$ is a small subinterval with high probability, and focusing on a small interval that receives most of the requests allows us to be more demanding when matching requests.

The algorithm then initializes the buckets $B_1, \ldots, B_b$ by partitioning  $T$ into $b$ sets that group nearby requests together and whose sizes are each $|T|/b$. Thus, this algorithm is \textit{ordinal}: it only considers the relative positions of requests, not their cardinal locations. For clarity, throughout this section we assume that there is a strict ordering across all agents on the line; for agents with the same location, this can be achieved either using a deterministic universal tie-breaking rule (e.g., based on the identities of the agents) or by choosing how to tie-break uniformly at random.

\begin{algorithm}[H]
\begin{algorithmic}
    \caption{The $\textsc{FindInitialBuckets}(T ,b)$ subroutine for the unknown-iid setting}
    \label{alg:FIBordinal}
    \State \textbf{Input:} set of sorted requests $T$, number of buckets $b$
   
    \State Partition $T$ into $b$ buckets $B_i = \left\{r_{(i-1)\cdot \frac{\learnBuckets}{b}+1}, \dots, r_{i\cdot \frac{\learnBuckets}{b}}\right\}$
    \State Define $f(x) = \begin{aligned}\begin{cases}
        \leftBucket, \quad & x < \min(T)\\
        j, &x \in \left[\min B_{j}, \min B_{j+1}\right) (\left[\min B_b, \max B_b\right] \textrm{ for } j=b)\\
        \rightBucket, & x > \max(T)
    \end{cases}
    \end{aligned}$
    
    \State \Return $B_{\leftBucket} = \emptyset, B_1, \dots, B_b, B_{\rightBucket} = \emptyset, f$
\end{algorithmic}
\end{algorithm}

\subsection{Analysis of \textsc{FindInitialBuckets} for the Unknown-i.i.d. Setting}\label{sec:FIBanalysis}

Before proving that Algorithm~\ref{alg:FIBordinal} satisfies the conditions in Section~\ref{sec:framework}, we give a few preliminary lemmas and definitions that are used throughout. Omitted proofs can be found in \Cref{appendix: omitted proofs unknown iid}.

\subsubsection{Preliminary Lemmas and Definitions}\label{subsubsec: random order prelims}

Suppose that requests $\left\{r_1, \dots, r_n\right\}$ are drawn i.i.d. from a distribution $\cD$ that is unknown to the algorithm. At time $t$, request $r_t$ is revealed to the algorithm. For ease of analysis, we consider the harder case where a set of ordered statistics $\left\{r_{\left(1\,:\,n\right)}, \dots, r_{\left(n\,:\,n\right)}\right\}$ are revealed in random order to the algorithm. This allows us to use known results from sampling without replacement. Note that this setting is harder because it allows for correlation between the locations of requests, whereas the i.i.d. condition does not. 

Given a set of requests $\{r_1, \dots, r_m\}$, and some $i \leq j \leq m$, we define $r_{(i:j)}$ to be the $i$-th leftmost of the requests $\{r_1, \dots, r_j\}$. Given a sequence of ordered requests $\left\{r_{(1:n)}, \dots, r_{(n:n)}\right\}$, we let $\Pr[E]$ be the probability of an event $E$ given said sequence and an ordering in which requests arrive $\pi \sim \Pi$, where $\Pi$ is the uniform distribution over permutations of $[n]$.

\begin{definition}[Hypergeometric Distribution]
  A random variable $X$ is said to be distributed according to the hypergeometric distribution with parameters $N,K$, and $m$ if, for $k \in [K]$, we have $\Pr[X = k] = \frac{\binom{K}{k} \binom{N-K}{m-k}}{\binom{N}{m}}$. 
\end{definition}

The classic formulation in which the hypergeometric distribution arises is the following~\cite{nelsonunivariate2005}: suppose we have an urn that contains $N$ balls, $K$ of which are colored red. Now, imagine drawing $m$ balls from the urn. Then, $X\sim \textrm{Hypergeometric}(N,K,m)$ is the random variable representing the number of red balls drawn.    

\begin{definition}[Negative Hypergeometric Distribution]\label{def: NHG dist.}
    A random variable $X$ is said to be distributed according to the negative hypergeometric distribution with parameters $N,K$, and $\ell$ if, for $k \in [K]$, we have $\Pr[X = k] = \frac{\binom{k+\ell-1}{k} \binom{N-\ell - k}{K-k}}{\binom{N}{K}}$. 
\end{definition}

The classic formulation in which the negative hypergeometric distribution arises is the following: suppose we have an urn that contains $N$ balls, $N-K$ of which are colored red. Now, imagine drawing balls from the urn until exactly $\ell$ red balls have been drawn. Then, $X\sim \textrm{NegativeHypergeometric}(N,K,\ell)$ is the random variable representing the number of balls drawn that are not red.

We next give the following tail bounds for hypergeometric random variables.

\begin{lemma}\label{lemma: hypergeometric tail bounds}[\cite{hoeffding1963probability}]
Let $X \sim \textrm{Hypergeometric}(n, k, m)$ and $\mu = \frac{k}{n} \cdot m$. Then, we have
\begin{alignat*}{2}
&\Pr\left[X \geq (1+\delta) \mu\right] \leq \exp\left(-\frac{\delta^2 \mu}{2+\delta}\right), \quad & 0 \leq \delta\\
&\Pr\left[X  \leq (1-\delta) \mu\right] \leq \exp\left(-\frac{\delta^2 \mu}{2}\right), & 0 \leq \delta \leq 1\\
&\Pr\left[\left|X-\mu\right| \geq \delta \mu\right] \leq 2\exp\left(-\frac{\delta^2 \mu}{3}\right), & 0 \leq \delta \leq 1
\end{alignat*}

\end{lemma}

The following lemma relates the \textit{negative} hypergeometric distribution to the hypergeometric distribution. The negative hypergeometric distribution is an important tool we will use in later sections, and relating it to the hypergeometric distribution allows us to inherit known tail bounds.

\begin{lemma}\label{lemma: relating negative hypergeometric to hypergeometric}
Let $N,K,m,\ell$ be arbitrary with $m \leq K \leq N$ and $\ell \leq N-K$. Then,
\[
\Pr\left[Y \sim \textrm{NegativeHypergeometric}(N, K, \ell) \leq m\right]= \Pr\left[X \sim \textrm{Hypergeometric}(N, N-K, m+\ell) \geq \ell\right]
\]
\end{lemma}
\begin{proof}
Let $A$ and $B$ be sets with sizes $K$ and $N-K$ respectively. Suppose that we draw elements from $A \cup B$ without replacement in two processes. In the first process, elements are drawn until exactly $\ell$ elements have been drawn from the set $B$. Let $Y$ be the number of elements drawn from $A$ during this process. Then, $Y \sim \textrm{NegativeHypergeometric}(N, K, \ell)$. In the subsequent process, we continue drawing elements until exactly $m+\ell$ elements have been drawn from $A \cup B$ (if there at least $m+\ell$ draws in the initial process, no elements will be drawn in the second process). Let $X$ be the number of elements drawn from $B$ during the first $m+\ell$ draws. Then, $X \sim \textrm{Hypergeometric}(N, N-K, m+\ell)$. 

We will show that $Y \leq m \iff X \geq \ell$. For the forward implication, suppose that $Y \leq m$. Then, the initial process contains at most $m+\ell$ draws, with $\ell$ of them being from the set $B$. Therefore, $X \geq \ell$. For the reverse direction, suppose that $X \geq \ell$. Then, at most $m$ items were drawn from $A$ across both processes. Since $Y$ is the number of elements drawn from $A$ in the first process, we clearly have $Y \leq m$. So, we conclude that $Y \leq m \iff X \geq \ell$. Since $Y \sim \textrm{NegativeHypergeometric}(N, K, \ell)$ and $X \sim \textrm{Hypergeometric}(N, N-K, m+\ell)$, we have
\[
\Pr\left[Y \sim \textrm{NegativeHypergeometric}(N, K, \ell) \leq m\right]= \Pr\left[X \sim \textrm{Hypergeometric}(N, N-K, m+\ell) \geq \ell\right]
\]
\end{proof}

Fix an arbitrary random order instance $\left(\left\{r_{(1 \, : n)}, \dots, r_{(n \, : n)}\right\}, \pi \right)$. Consider the buckets created during the learning phase. For $j \in \left[\beta^{\rho}\right]$, let $\ell_j = \argmax_{\ell} \left\{\sum_{i = 1}^{\ell} \mathbbm{1}_{\pi(i) \leq \learnBuckets} \leq j \cdot \frac{\beta}{8}+1\right\}$ be the index of the ordered statistic that corresponds to the right endpoint of the $j$-th bucket (let $\ell_0 = \argmin_i \left\{\pi(i) \leq \learnBuckets\right\}$ be the smallest index corresponding to a learning phase request). Then, let $Y_j = \sum_{i=\ell_{j-1}}^{\ell_j} \left(1 - \mathbbm{1}_{\pi(i) \leq \learnBuckets}\right)$ be the number of non-learning phase requests that are within the endpoints of the $j$-th bucket.

\begin{lemma}\label{lemma: number between buckets is NHG random variable}
    For any $j \in \left[\beta^{\rho}\right]$, we have that $Y_j \sim \textrm{NegativeHypergeometric}\left(n, n-\learnBuckets, \beta/8\right)$.
\end{lemma}

\begin{proof}
Recall that, because the requests appear in random order, the requests that arrive in the learning phase constitute a uniformly random subset of size $\learnBuckets$ of the total set of requests. Thus, each $Y_j$ follows the same distribution. Let the set of learning phase requests be $T$, and the set of non-learning phase requests be $[n]\setminus T$. Then, each $Y_j$ corresponds to the number of elements drawn from $[n]\setminus T$ before $\beta/8$ elements are drawn from $T$ when elements from $[n]$ are drawn without replacement. This is the same formulation described in \Cref{def: NHG dist.} with parameters $n$, $n-\learnBuckets$, and $\beta/8$. So, we conclude $Y_j \sim \textrm{NegativeHypergeometric}(n,n-\learnBuckets, \beta/8)$. 
\end{proof}

The next lemma uses the relation between the hypergeometric distribution and the \textit{negative} hypergeometric distribution, along with tail bounds for the hypergeometric distribution, to tail bound the \textit{negative} hypergeometric distribution.
\begin{lemma}\label{lemma: negative hypergeometric tail bounds}
    For $Y \sim \textrm{NegativeHypergeometric}\left(n, n-\learnBuckets, \beta/8\right)$, we have 
    \[
    \Pr\left[Y\leq \frac{n\beta}{16\learnBuckets} \right] \leq \frac{1}{n^2} \quad \textrm{and} \quad \Pr\left[Y \geq \frac{n\beta}{4\learnBuckets} \right] \leq \frac{1}{n^2}.
    \]
\end{lemma}

\begin{proof}[Proof sketch]
    The lemma is proved by using \Cref{lemma: relating negative hypergeometric to hypergeometric} to rewrite the tail bounds for the negative hypergeometric distribution in terms of tail bounds for the hypergeometric distribution. The result then follows from \Cref{lemma: hypergeometric tail bounds}.
\end{proof}

Let $E_M = \bigcap_{j \in \left[\beta^{\rho}\right]} \left\{Y_j \in \left[\frac{n\beta}{16\learnBuckets}, \frac{n\beta}{4\learnBuckets} \right]\right\}$ be the event that, between the endpoints of any bucket, there are at least $\frac{n\beta}{16\learnBuckets}$ non-learning phase requests and at most $\frac{n\beta}{4\learnBuckets}$ non-learning phase requests.

\begin{lemma}\label{lemma: enough total requests in buckets - random order}
    We have $\Pr\left[E_M \right] \geq 1-\frac{1}{n^2}$. 
\end{lemma}

\begin{proof}
    Note that:
    \begin{align*}
        \Pr\left[\bar{E}_M \right] = \Pr\left[\bigcup_{j \in \left[\beta^{\rho}\right]} \left\{Y_j \not\in \left[\frac{n\beta}{16\learnBuckets}, \frac{n\beta}{4\learnBuckets}\right] \right\}\right] &\leq\sum_{j \in \left[\beta^{\rho}\right]} \Pr\left[Y_j \not \in \left[\frac{n\beta}{16\learnBuckets}, \frac{n\beta}{4\learnBuckets}\right]\right].
    \end{align*}
    So, for any $j \in \left[\beta^{\rho}\right]$, it suffices to bound $\Pr\left[Y_j \not \in \left[\frac{n\beta}{16\learnBuckets}, \frac{n\beta}{4\learnBuckets}\right]\right]$. Note that by Lemma~\ref{lemma: number between buckets is NHG random variable}, we have $Y_j \sim \textrm{NegativeHypergeometric}\left(n, n-\learnBuckets, \beta/8\right)$. 
    Then: 
    \begin{align*}
        \Pr\left[Y_j \not \in \left[\frac{n\beta}{16\learnBuckets}, \frac{n\beta}{4\learnBuckets}\right] \right] = \Pr\left[Y_j \sim \textrm{NegativeHypergeometric}\left(n, n-\learnBuckets, \beta/8\right) \not \in \left[\frac{n\beta}{16\learnBuckets}, \frac{n\beta}{4\learnBuckets}\right] \right] \leq \frac{2}{n^3}, \tag{1}
    \end{align*}
    where the inequality follows from Lemma~\ref{lemma: negative hypergeometric tail bounds}. 
    We conclude that:
    \begin{align*}
        \Pr\left[\bar{E}_M \right] = \Pr\left[\bigcup_{j \in \left[\beta^{\rho}\right]} \left\{Y_j \not\in \left[\frac{n\beta}{16\learnBuckets}, \frac{n\beta}{4\learnBuckets}\right] \right\}\right] &\leq\sum_{j \in \left[\beta^{\rho}\right]} \Pr\left[Y_j \not \in \left[\frac{n\beta}{16\learnBuckets}, \frac{n\beta}{4\learnBuckets}\right]\right] \leq \beta^{\rho} \cdot \frac{2}{n^3} < \frac{1}{n^2},
    \end{align*}
    where the first inequality is due to (1) and the second inequality follows from $\beta^{\rho} < n/2$.
\end{proof}

\subsubsection{Algorithm~\ref{alg:FIBordinal} satisfies condition~\ref{cond: probability of hitting wall}}\label{subsubsec: FIBordinal satisfies condition 1}

For $k \in [\rho]$, let $X^k_j = \left|\left\{t \in \left\{t_{k-1}+1, \dots, t_k\right\} : i(t) = j \right\}\right|$ be the number of requests in the $j$-th bucket in the $k$-th phase. Also, let $F = \cap_{k \in [\rho]}\cap_{j \in \left[\beta^{\rho+1-k}\right]} \left\{X^k_j \geq \beta\right\}$ be the event that every bucket in every phase receives at least $\beta$ requests. The below lemma shows that this event occurs with high probability.

\begin{lemma}\label{lemma: enough requests in each bucket - random order}
    Algorithm~\ref{alg:framework} equipped with the \textsc{FindInitialBuckets} subroutine defined by Algorithm~\ref{alg:FIBordinal} is such that, for any sequence of ordered requests $\left\{r_{(1)}, \dots, r_{(n)}\right\}$, $\Pr[F] \geq 1-2/n$.
\end{lemma}

\begin{proof}
Let $k \in [\rho], j \in \left[\beta^{\rho+1-k}\right]$ be arbitrary. By the law of total probability, we have
\begin{align*}
\Pr\left[X^k_j < \beta\right] &= \Pr\left[X^k_j < \beta \middle| E_M\right] \cdot \Pr\left[E_M\right] + \Pr\left[X^k_j < \beta \middle| \bar{E}_M\right] \cdot \Pr\left[\bar{E}_M\right]\\
&\leq \Pr\left[X^k_j < \beta \middle| E_M\right] + \Pr\left[\bar{E}_M\right]\\
&\leq \Pr\left[X^k_j < \beta \middle| E_M\right] + \frac{1}{n^2},
\end{align*}
where the last line follows from Lemma~\ref{lemma: enough total requests in buckets - random order}. Now, recall that the interval mapping to $B_j$ in the $k$-th phase is $\left(r_{\left((j-1)\cdot \frac{\beta^k}{8} \, : \, \learnBuckets\right)}, r_{\left(j\cdot \frac{\beta^k}{8} \, : \, \learnBuckets\right)}\right]$, since $B_j$ in the $k$-th phase is formed from $\beta^{k-1}$ buckets from the first phase, each of which received $\beta/8$ requests in the learning phase. Let $\cB_j = \left\{j' \cdot \frac{\beta}{8} : (j-1) \cdot \frac{\beta^k}{8} \leq j' \cdot \frac{\beta}{8}  \leq j \cdot \frac{\beta^k}{8}\right\}$ be the set of indices of initial buckets that are combined to form the $j$-th bucket in the $k$-th phase. Recall that for $j' \in \left[\beta^{\rho}\right]$, $Y_{j'}$ is the number of non-learning phase requests between the endpoints of the $j'$-th initial bucket. Therefore, the number of non-learning phase requests between the endpoints of $j$-th bucket in the $k$-th phase is $\sum_{j' \in \cB_j} Y_{j'}$. Note that these $Y_{j'}$ are fixed after the learning phase. Further, conditioned on $E_M$, each $Y_{j'}$ is at least $\frac{n \beta}{16\learnBuckets}$. Thus, conditioned on $E_M$, $\sum_{j' \in \cB_j} Y_{j'} \geq \frac{n \beta}{16\learnBuckets} \cdot \beta^{k-1} = \frac{n \beta^k}{16\learnBuckets}$. 

Now, given that the $Y_{j'}$ are fixed before the $k$-th phase, recall that the number of requests from the $k$-th phase that appear in the $j$-th bucket will be some fraction of $\sum_{j' \in \cB_j} Y_{j'}$. Since the requests arrive in random order, and there are $n-\learnBuckets$ \textit{total} non-learning phase requests and $t_k - t_{k-1}$ requests in the $k$-th phase, the fraction of the $\sum_{j' \in \cB_j} Y_{j'}$ requests that appear in the $k$-th phase is distributed according to the hypergeometric distribution with parameters $n-\learnBuckets$, $\sum_{j' \in \cB_j} Y_{j'}$, and $t_k - t_{k-1}$. So, conditioned on the $Y_{j'}$ being fixed before the $k$-th phase, $X^k_j \sim \textrm{Hypergeometric}\left(n-\learnBuckets, \sum_{j' \in \cB_j} Y_{j'}, t_k - t_{k-1}\right)$.

Therefore,
\begin{align*}    
\Pr\left[X^k_j < \beta \middle| E_M\right] &= \Pr\left[X^k_j \sim \textrm{Hypergeometric}\left(n-\learnBuckets, \sum_{j' \in \cB_j} Y_{j'}, t_k - t_{k-1}\right) < \beta\middle| E_M\right]\\
&\leq \Pr\left[X^k_j \sim \textrm{Hypergeometric}\left(n-\learnBuckets, \frac{n\beta^{k}}{16\learnBuckets}, t_k - t_{k-1}\right) < \beta\right],
\end{align*}
where the second line follows from the fact that, for $K \geq K'$, stochastic dominance gives 
\[
\Pr\left[X \sim \textrm{Hypergeometric}(N,K,m) \leq \ell\right] \leq \Pr\left[X \sim \textrm{Hypergeometric}(N,K',m) \leq \ell\right] .
\]

Then, hypergeometric tail bounds described in Lemma~\ref{lemma: hypergeometric tail bounds} allow us to bound this probability: let
\begin{align*}
\mu = \frac{t_k - t_{k-1}}{n-\learnBuckets} \cdot \frac{n\beta^k }{16\learnBuckets} > \frac{t_k - t_{k-1}}{16\learnBuckets} \ \cdot \beta^k \geq \frac{3\beta^{\rho+1-k}(\beta-1)}{16 \cdot \frac{\beta^{\rho+1}}{8}} \cdot \beta^k =\frac{3}{2}(\beta-1)\geq \frac{5}{4}\beta,
\end{align*}
where the second inequality follows from $n \geq 4\beta^{\rho+1}$. Additionally, let $\delta = 1-\frac{\beta}{\mu} \geq 1 - \frac{\beta}{\frac{5}{4}\beta} = \frac{1}{5}$. Then, we have
\allowdisplaybreaks
\begin{align*}
    &\Pr\left[X \sim \textrm{Hypergeometric}\left(n-\learnBuckets, \frac{\beta^k n}{16\learnBuckets}, t_k - t_{k-1}\right) < \beta\right]\\
    = \ &\Pr\left[X \sim \textrm{Hypergeometric}\left(n-\learnBuckets, \frac{\beta^k n}{16\learnBuckets}, t_k - t_{k-1} \right) < (1-\delta) \cdot \mu \right]\\
    \leq \ &\exp\left(-\frac{\delta^2 \mu}{2}\right)\\ 
    \leq \ &\exp\left(-\frac{\mu}{50}\right)\\
    \leq \ &\exp\left(-\frac{\beta}{40}\right)\\
    \leq \ &\frac{1}{n^2},
\end{align*}
where the third line follows from Lemma~\ref{lemma: hypergeometric tail bounds}, the fourth line follows from $\delta \geq \frac{1}{5}$, the fifth line follows from $\mu \geq \frac{5}{4}\beta$, and the last line follows from from $\beta \geq 500 \log n$. So, 
\begin{align*}
\Pr\left[X^k_j < \beta\right] &\leq \Pr\left[X^k_j \sim \textrm{Hypergeometric}\left(n-\learnBuckets, \frac{\beta^k n}{16\learnBuckets}, t_k - t_{k-1}\right) < \beta \middle| E_M \right] + \Pr\left[\bar{E}_M\right]\\
&\leq \frac{1}{n^2} + \frac{1}{n^2} = \frac{2}{n^2}.
\end{align*}

Therefore, we conclude
\begin{align*}
    \Pr[F] 
    &= 1 - \Pr\left[\cup_{k \in [\rho]}\cup_{j \in \left[\beta^{\rho+1-k}\right]} \left\{X^k_j < \beta\right\}\right] \\
    &\geq 1-\sum_{k \in [\rho]} \sum_{j \in \left[\beta^{\rho+1-k}\right]} \Pr\left[X^k_j < \beta\right] \\ 
    &\geq 1 - \frac{2}{n^2} \sum_{k \in [\rho]} \beta^{\rho+1-k}\\
    &\geq 1 - \frac{2}{n^2} \beta^{\rho+1}\\
    &\geq 1-\frac{2}{n}.
\end{align*}
\end{proof}

Combining the above, we conclude that we are not forced to make arbitrary assignments.
\begin{lemma}\label{lemma: high probability of not hitting wall - random order}
    Algorithm~\ref{alg:framework} equipped with the \textsc{FindInitialBuckets} subroutine defined by Algorithm~\ref{alg:FIBordinal} is such that $\Pr[A_{\rho}] \geq 1-2/n$. 
\end{lemma}

\begin{proof}
First, note that $\Pr[A_{\rho}] \geq \Pr\left[F\right]$. This follows from the fact that the number of future requests at the end of a phase is equal to three times the number of buckets, the fact that there are at most $2$ requests outside of the buckets, and $F$ which guarantees that, in each phase, each bucket receives at least as many requests as it began the phase with, so no bucket has more than one unmatched request. Therefore, we have 
\begin{align*}
    \Pr\left[A_{\rho}\right] \geq \Pr\left[F\right] \geq 1-2/n,
\end{align*}
where the last inequality follows from Lemma~\ref{lemma: enough requests in each bucket - random order}.
\end{proof}

\subsubsection{Algorithm~\ref{alg:FIBordinal} satisfies condition~\ref{cond: few arrivals in any one bucket}}\label{subsubsec: FIBordinal satisfies condition 2}

The next two lemmas prove the first part of the condition, which bounds the probability that more than $\beta^2$ requests arrive outside of the buckets. 

\begin{lemma}
\label{lemma: hypergeometric fact}
    Let $r_{(1:n)} \leq r_{(2:n)} \leq \dots \leq r_{(n:n)}$ be the order statistics of instance $I = (r_1, \dots, r_n)$. 
    Then $\Pr_{I \sim \cD^n}\left[ \min_{t \in [k]}r_t > r^{(m)} \right] = \Pr_{I \sim \cD^n}\left[ \max_{t \in [k]}r_t < r^{(n - m)} \right] \leq e^{- km/n}$.
\end{lemma}

\begin{lemma}
\label{lemmma: few requests outside of buckets - random order} We have that $\Pr\left[\left| \{t \in \{ \learnBuckets+1, \dots, n\}: i(t) \in \left\{\leftBucket, \rightBucket\right\} \} \right| \leq 32\beta^2\right] \geq 1 - \frac{1}{n}$.
\end{lemma}

\begin{proof}
    Let $r_{(1\,: \, n)} \leq \dots \leq r_{(n\,: \, n)}$  be the order statistics of the $n$ samples. Then:
    \begin{align*}
        &\Pr\left[ |\{t \in \{ \learnBuckets+1, \dots, n\}: i(t) \in \{\leftBucket, \rightBucket\} \}| \leq 32\beta^2 \right]  \\
        \geq_{(1)} \ &\Pr\left[ |\{t \in \{ \learnBuckets+1, \dots, n\}: i(t) \in \{\leftBucket, \rightBucket\} \}| \leq \frac{n\beta}{\learnBuckets} \right]  \\
        \geq \ &\Pr\left[ \min_{t \in [\learnBuckets]}r_t \leq r_{\left(\frac{n\beta}{2\learnBuckets}\,: \, n\right)}, \max_{t \in [\learnBuckets]}r_t \geq r_{\left(n - \frac{n\beta}{2\learnBuckets}\,: \, n\right)} \right] \\ 
        \geq \ &1 - \Pr\left[ \min_{t \in [\learnBuckets]}r_t > r_{\left(\frac{n\beta}{2\learnBuckets}\,: \, n\right)}\right] - \Pr\left[ \max_{t \in [\learnBuckets]}r_t < r_{\left(n - \frac{n\beta}{2\learnBuckets}\,: \, n\right)}\right] \\
        \geq_{(2)} \ &1 - 2\exp\left(- \frac{\learnBuckets \cdot \frac{n\beta}{2\learnBuckets}}{n} \right)\\
        = \ &1- 2\exp\left(-\beta/2\right)\\
        \geq_{(3)} \ &1  - \frac{2}{n},
    \end{align*}
    where (1) follows from $n \leq 4\beta^{\rho+2}$ and $\learnBuckets = \frac{\beta^{\rho+1}}{8}$ implies $\frac{n}{\learnBuckets} \leq 32\beta$, (2) is due to \cref{lemma: hypergeometric fact} for $k = \learnBuckets$ and $m = \frac{n \beta}{2\learnBuckets}$, and (3) from $\beta = \left\lfloor 500 \log n\right\rfloor$ and $n \geq 2$.
\end{proof}

The second and third parts of the condition are proved in the below lemma.
\begin{lemma}\label{lemma: bound on expected number any bucket receives - random order}
    For buckets in the first phase, we have  
    \[
    \Pr\left[ \forall j \in \left[\beta^{\rho}\right]: \left| \{t \in \{\learnBuckets+1, \dots, t_{1}\}: i(t) = j \} \right| \leq 32\beta^2 \right] \geq 1-1/n,
    \]
    and for buckets in subsequent phases, we have 
    \[
    \Pr[ \forall k \in [2..\rho], j \in \left[\beta^{\rho + 1 - k}\right]: \left| \{t \in \{ t_{k-1}+1, \dots, t_{k}\}: i(t) = j \} \right| \leq 24\beta ] \geq 1-1/n.
    \]
\end{lemma}

\begin{proof}[Proof sketch] As in the proof of \Cref{lemma: enough requests in each bucket - random order}, we use \Cref{lemma: number between buckets is NHG random variable} to give a bound on the expected number of requests in each bucket. Specifically, we show, $E\left[X^k_j\right] \leq \frac{1}{n} + 
\begin{aligned}
    \begin{cases}
        \frac{4n}{\beta^{\rho}} , \quad &\textrm{ if } k=1\\
        12\beta , & \textrm{else}
    \end{cases}
\end{aligned}$. We then apply a Chernoff bound to get the high probability result.
    
\end{proof}

\subsubsection{Algorithm~\ref{alg:FIBordinal} satisfies condition~\ref{cond: buckets have bounded length}}\label{subsubsec: FIBordinal satisfies condition 3}

\begin{lemma}
\label{lemma: bounded bucket length - random order}
We have that $\E_{I \sim \cD^n}\left[\sum_{j \in \left[\beta^{\rho}\right]} (\sup{I^1_j} - \inf{I^1_j})\right] 
    \leq \E_{I \sim \cD^n}\left[ \max_{t \in [n]}r_{t} - \min_{t \in [n]}r_{t} \right]$.
\end{lemma}

\begin{proof}
    For any instance $I = \left(\left\{r_{(1)}, \dots, r_{(n)}\right\}, \pi\right)$, 
    \begin{align*}
        \sum_{j \in \left[\beta^{\rho}\right]} (\sup{I^1_j} - \inf{I^1_j}) = r_{(\learnBuckets: \learnBuckets)} - r_{(1: \learnBuckets)} \leq r_{(n:n)} - r_{(1:n)} = \max_t r_t - \min_t r_t
    \end{align*}
\end{proof}

\subsubsection{Putting everything together}

By combining the results from this section with the main lemma for the framework, we obtain our main result.
\UnknownIIDRatio*

\begin{proof}
  By Lemma~\ref{lem: ratio of algorithmic framework}, it suffices to show that Algorithm~\ref{alg:framework}, with the subroutine \textsc{FindIntialBuckets} defined by Algorithm~\ref{alg:FIBordinal} and with learning parameter $\learnBuckets = \frac{\beta^{\rho+1}}{8}$, satisfies \Cref{cond: probability of hitting wall,cond: few arrivals in any one bucket,cond: buckets have bounded length}. \Cref{cond: probability of hitting wall} is shown by \cref{lemma: high probability of not hitting wall - random order}, \Cref{cond: few arrivals in any one bucket} by \cref{lemmma: few requests outside of buckets - random order,lemma: bound on expected number any bucket receives - random order}, and \Cref{cond: buckets have bounded length} by \cref{lemma: bounded bucket length - random order}.
\end{proof}

\section{Random Order Impossibility Result}\label{sec:random_order}

This section focuses on the impossibility result for online min-cost matching on the line in the random order model.

\RandomOrderImpossibility*

The proof follows by constructing two instances $I_1, I_2$ that differ in only two locations of requests, such that the cost of any matching on them varies widely based on the decisions made for matching these two locations.
Then we argue that any algorithm, during its online runtime, is incapable of distinguishing between $I_1, I_2$, and as a result cannot achieve a cost bounded by the optimal cost on both.
For the purposes of the proof, we denote by $x$-request, a request that is located in $x \in [0,1]$.
\begin{proof}
    For any $n \geq 4, \epsilon \in (0,1)$, we define the following multi-sets, both of cardinality $n$:
    \begin{align*}
        I_1 = \{ \underbrace{0, \dots, 0}_{n-4}, \frac{1}{2}, \frac{1}{2}, \frac{1}{2} + \epsilon, \frac{1}{2} + \epsilon \}, \quad
        I_2 = \{ \underbrace{0, \dots, 0}_{n-4}, \frac{1}{2}, \frac{1}{2} + \epsilon, 1-\epsilon, 1 \}.
    \end{align*}
    We show the following two points, from which the statement follows:
    \begin{enumerate}
        \item On instance $I_1$, we show that any randomized algorithm with bounded competitive ratio must match $\left(\frac{1}{2} + \epsilon \right)$-requests to $\frac{1}{2}$-requests with probability $0$,
        \item Any algorithm that satisfies (1) has an unbounded competitive ratio on instance $I_2$.
    \end{enumerate}
    For any random permutation $\pi(I) = (r_1, r_2, \dots, r_n)$, we define as $M_t(\pi(I))$ the matching after time step $t$ of the algorithm. 
    We define the event that two requests in locations $\frac{1}{2}, \frac{1}{2} + \epsilon$ arrive before other non $0$-requests have arrived: 
    \begin{align*}
        E_{order}(\pi(I)) = \{ \exists t_1, t_2 \in [n]: r_{t_1} = \frac{1}{2}, r_{t_2} = \frac{1}{2} + \epsilon, r_{t'} = 0 \ \forall t' \in [\max(t_1, t_2)] \setminus \{t_1, t_2\} \}.
    \end{align*}
    We also define the event that these two requests are matched to each other:
    \begin{align*}
        E_{match}(\pi(I)) = \{\exists t_1, t_2 \in [n]:  r_{t_1} = \frac{1}{2}, r_{t_2} = \frac{1}{2} + \epsilon, r_{t'} = 0 \forall t' \in [\max(t_1, t_2)] \setminus \{t_1, t_2\}, \\ \{r_{t_1}, r_{t_2}\} \in M_{\max(t_1, t_2)}(\pi(I))\}.
    \end{align*}
    We start with (1).
    Note that by \cref{lemma: OPT structure} the optimal matching matches all the $0$-requests to each other (since $n-4$ is even), the $\frac{1}{2}$-requests to each other, and the $\left( \frac{1}{2} + \epsilon \right)$-requests to each other, therefore $\E_{\pi}\left[ \opt(\pi(I_1)) \right] = 0$.
    
    Let $p = \Pr_{\pi, \alg}\left[ E_{match}(\pi(I_1)) \middle| E_{order}(\pi(I_1)) \right]$, then the algorithm finds a matching with the following cost:
    \begin{align*}
        &\E_{\pi, \alg}\left[ \alg(\pi(I_1)) \right]  \\
       \geq \ &\E_{\pi, \alg}\left[ \alg(\pi(I_1)) \middle| E_{order}(\pi(I_1)) \right] \cdot \Pr_{\pi, \alg}\left[ E_{order}(\pi(I_1)) \right] \\
        \geq \ &\E_{\pi, \alg}\left[ \alg(\pi(I_1)) \middle| E_{order}(\pi(I_1)) \right] \cdot \frac{1}{6} \\
        \geq \ &\E_{\pi, \alg}\left[ \alg(\pi(I_1)) \middle| E_{order}(\pi(I_1)), E_{match}(\pi(I_1)) \right] \cdot \Pr_{\pi, \alg}\left[ E_{match}(\pi(I_1)) \middle| E_{order}(\pi(I_1))\right] \cdot \frac{1}{6} \\
        \geq \ &\frac{\epsilon \cdot p}{6},
    \end{align*}
    where the first line is by tower rule, the second line is because by definition of event $E_{order}(\pi(I_1))$, it only happens one $\frac{1}{2}$-request and one $\left(\frac{1}{2} + \epsilon\right)$-request arrive before the other two such requests which happens with probability at least $\frac{1}{6}$ by considering all possible orders of the non $0$-requests, the third line is by tower rule and the fourth line is because if a $\frac{1}{2}$-request is matched to a $\left( \frac{1}{2} + \epsilon \right)$-request, then $\alg$ pays at least the distance $\epsilon$. 
    Assume $p > 0$, then $\E_{\pi, \alg}\left[ \alg(\pi(I_1)) \right] > 0$, therefore the competitive ratio is unbounded. 
    
    We proceed to show (2) by assuming $p = 0$.
    Note that by \cref{lemma: OPT structure}, the optimal matching matches all $0$-requests to each other (since $n-4$ is even), the $\frac{1}{2}$-request is matched to the $\left( \frac{1}{2} + \epsilon \right)$-request and the $(1-\epsilon)$-request is matched to the $1$-request, which gives $\E_{\pi}\left[ \opt(\pi(I_2)) \right] = 2 \epsilon$.
    We have that:
    \begin{align*}
        &\E_{\pi, \alg}\left[ \alg(\pi(I_2)) \right] \\
        \geq \ &\E_{\pi, \alg}\left[ \alg(\pi(I_2)) \middle| E_{order}(\pi(I_2)) \right] \cdot \Pr_{\pi, \alg}\left[ E_{order}(\pi(I_2)) \right] \\
        \geq \ &\E_{\pi, \alg}\left[ \alg(\pi(I_2)) \middle| E_{order}(\pi(I_2)) \right] \cdot \frac{1}{6} \\
        \geq \ &\E_{\pi, \alg}\left[ \alg(\pi(I_2)) \middle| E_{order}(\pi(I_2)), \bar{E}_{match}(\pi(I_2)) \right] \cdot \Pr_{\pi, \alg}\left[ \bar{E}_{match}(\pi(I_2)) \middle| E_{order}(\pi(I_2)) \right] \cdot \frac{1}{6} \\
        = \ &\E_{\pi, \alg}\left[ \alg(\pi(I_2)) \middle| E_{order}(\pi(I_2)), \bar{E}_{match}(\pi(I_2)) \right] \cdot \frac{1}{6} \\
        = \ &\frac{1/2-\epsilon}{6} \\
        = \ &\frac{1/2-\epsilon}{12 \epsilon} \cdot \E_{\pi}\left[ \opt(\pi(I_2)) \right],
    \end{align*}
    where the first line is by tower rule, the second line is because the $\frac{1}{2}$, $\left(\frac{1}{2} + \epsilon\right)$-requests arrive before $(1-\epsilon), 1$-requests in $\pi(I_2)$ with probability $\frac{1}{6}$ by reordering, the third line is by tower rule, the fourth line is because $\Pr_{\pi, \alg}\left[ \bar{E}_{match}(\pi(I_2)) \middle| E_{order}(\pi(I_2)) \right] = \Pr_{\pi, \alg}\left[ \bar{E}_{match}(\pi(I_1)) \middle| E_{order}(\pi(I_1)) \right] = 1 - p = 1$ because at the arrival of the $\frac{1}{2}$, $\left( \frac{1}{2} + \epsilon \right)$-requests in $\pi(I_2)$ we have no information that distinguishes it from $\pi(I_1)$, therefore the algorithm matches $\frac{1}{2}$ to $\left( \frac{1}{2} + \epsilon \right)$ with the same probability and by assumption $p = 0$, the fifth line is because if the $\frac{1}{2}$-request is not matched to the $\left( \frac{1}{2} + \epsilon \right)$-request, then it will be matched to either the $1-\epsilon$ or the $1$-request and the cost will be at least $\frac{1}{2} - \epsilon$, and the sixth line due to $\E_{\pi}\left[ \opt(\pi(I_2)) \right] = 2\epsilon$.
    
    Then the ratio is unbounded as $\epsilon \rightarrow 0$.
\end{proof}

\section{Bipartite Case with Two-sided Arrivals}\label{sec:bipartite}

In this section, we focus on bipartite online matching with general arrivals under known-i.i.d. arrivals.
We refer to the two sides as left requests and right requests, and both are sampled i.i.d. from distribution $\cD$.
Our main result for this section is \cref{thm: ratio for two sided bipartite arrivals}.
\TwoSidedRatio*

We define the reduction in \cref{alg:reduction from bipartite two-sided arrivals to bipartite one-sided}.

\subsection{The Algorithm}
\Cref{alg:reduction from bipartite two-sided arrivals to bipartite one-sided} does not match any right request to any left request, until the number of total requests that have arrived equals $n$. 
Then, it uses $\cA$ to match the incoming set of right requests with the already arrived left requests and the incoming set of left requests with the already arrived right requests.
\begin{algorithm}
\begin{algorithmic}
\caption{OneSidedReduction$(n, \cD, \cA)$} 
\label{alg:reduction from bipartite two-sided arrivals to bipartite one-sided}
    \State \textbf{Input:} number of requests $n$, distribution $\cD$ and online algorithm $\cA$
    \State $M_{S}, M_{R} \gets \emptyset$
    \State $S_1, R_1 \gets \emptyset$
    \While{$|S_1| + |R_1| < n$}
        \If{request $r$ arrived} 
            \State $R_1 \gets R_1 \cup \{r\}$
        \ElsIf{server $s$ arrived}
            \State $S_1 \gets S_1 \cup \{s\}$
        \EndIf
    \EndWhile
    \State $M_S \gets \cA$ for arrivals of right requests from distribution $\cD$ with servers $S_1$
    \State $M_R \gets \cA$ for arrivals of left requests from distribution $\cD$ with servers $R_1$
    \State \Return $M_S \cup M_R$
\end{algorithmic}
\end{algorithm}

\subsection{Analysis}
Before we state \cref{lem:reduction from one sided bipartite to two sided with different distributions}, we introduce notation for the cost of the algorithm and the optimal solution.
We refer to the cost incurred by \cref{alg:reduction from bipartite two-sided arrivals to bipartite one-sided} on ordered left requests $S$, ordered right requests $R$, and $\pi$, as $\osr(S, R, \pi) = \sum_{\{s, r\} \in M_S \cup M_R}d(s, r)$,
to the cost incurred by algorithm $\cA$ given left requests $S$ and ordered right requests $R$ as $\cA(S, R)$,
and to the cost of the optimal solution for any left requests $S$ and right requests $R$ as $\opt(S, R)$.

\begin{lemma}
\label{lem:reduction from one sided bipartite to two sided with different distributions}
    Assume that there exists an online algorithm $\cA$ such that given a set of servers $S$ and a distribution $\cD$ from which $|S|$ requests will be drawn sequentially, gives a perfect matching with cost $\E_{R \sim \cD^{|S|}}\left[ \cA(S, R) \right] \leq a \cdot \E_{R \sim \cD^{|S|}}\left[ \opt(S, R) \right]$.
    Then for all $\pi$, $\E_{R \sim \cD^n, S \sim \cD^n}\left[ \osr(S, R, \pi) \right] \leq 2a \cdot \E_{R \sim \cD^n, S \sim \cD^n}\left[ \opt(S, R) \right]$, where $\osr$ is \cref{alg:reduction from bipartite two-sided arrivals to bipartite one-sided}.
\end{lemma}

Before showing the proof of \cref{lem:reduction from one sided bipartite to two sided with different distributions}, we show that the expected value of the optimal solution is increasing in the number of requests $n$.
\begin{lemma}
\label{lem: monotonicity of bipartite optimal solution}
    Let $n \in \mathbbm{N}$. Then $\E_{S \sim \cD^n, R \sim \cD^n}\left[ \opt(S, R) \right] \leq \E_{S \sim \cD^{n+1}, R \sim \cD^{n+1}}\left[ \opt(S, R) \right]$.
\end{lemma}
\begin{proof}
    Let $S_i = \{s_1, s_2, \dots, s_i\}$ and $R_i = \{r_1, r_2, \dots, r_i\}$ for any $i \in [n+1]$.
    For any $x \in [0,1]$ we define the following random variable: 
    $$Z_n(x) = \sum_{i \in [n]}\mathbbm{1}\{s_i \leq x\} - \sum_{i \in [n]}\mathbbm{1}\{r_i \leq x\}.$$
    Notice that:
    \begin{align*}
        &\E_{s_{n+1} \sim \cD, r_{n+1} \sim \cD}\left[ Z_{n+1}(x) \middle| Z_n(x) \right] \\
        =\ &\E_{s_{n+1} \sim \cD, r_{n+1} \sim \cD}\left[ Z_{n}(x) + \mathbbm{1}\{s_i \leq x\} - \mathbbm{1}\{r_i \leq x\} \middle| Z_n(x) \right] \\
        =\ &Z_n(x) + \E_{s_{n+1} \sim \cD, r_{n+1} \sim \cD}\left[ \mathbbm{1}\{s_{n+1} \leq x\} - \mathbbm{1}\{r_{n+1} \leq x\} \middle| Z_n(x) \right] \\
        =\ &Z_n(x) + \E_{s_{n+1} \sim \cD, r_{n+1} \sim \cD}\left[ \mathbbm{1}\{s_{n+1} \leq x\} - \mathbbm{1}\{r_{n+1} \leq x\} \right] \\
        =\ &Z_n(x), \tag{1}
    \end{align*}
    where the first line is by definition of $Z_{n+1}(x)$, the second line id due to conditioning on $Z_n(x)$, the third line is because $s_{n+1}, r_{n+1}$ are independent of $S_n, R_n$ and therefore independent of $Z_n(x)$ and the fourth line is because $\E_{s_{n+1} \sim \cD}\left[ \mathbbm{1}\{s_{n+1} \leq x\} \right] = \Pr_{s_{n+1} \sim \cD}\left[ s_{n+1} \leq x \right] = \Pr_{r_{n+1} \sim \cD}\left[ r_{n+1} \leq x \right] = \E_{r_{n+1} \sim \cD}\left[ \mathbbm{1}\{r_{n+1} \leq x\} \right]$.
    This implies:
    \begin{align*}
        &\E_{S_{n+1} \sim \cD^{n+1}, R_{n+1} \sim \cD^{n+1}}\left[ |Z_{n+1}(x)| \right] \\
        \geq \ &| \E_{S_{n+1} \sim \cD^{n+1}, R_{n+1} \sim \cD^{n+1}}\left[ Z_{n+1}(x) \right] | \\
        =\ &| \E_{S_{n} \sim \cD^{n}, R_{n} \sim \cD^{n}}\left[ \E_{s_{n+1} \sim \cD, r_{n+1} \sim \cD}\left[ Z_{n+1}(x) \middle| Z_n(x) \right] \right] | \\
        =\ &| \E_{S_{n} \sim \cD^{n}, R_{n} \sim \cD^{n}}\left[ Z_n(x) \right] |, \tag{2}
    \end{align*}
    where the first line is by Jensen since $\phi(x) = |x|$ is convex, the second line is by tower rule and the third line is due to (1).
    Then we have:
    \begin{align*}
        &\E_{S_{n+1} \sim \cD^{n+1}, R_{n+1} \sim \cD^{n+1}}\left[ \opt(S_{n+1}, R_{n+1})\right] \\ 
        =\ &\E_{S_{n+1} \sim \cD^{n+1}, R_{n+1} \sim \cD^{n+1}}\left[ \int_0^1{|Z_{n+1}(x)|}\,dx \right] \\
        =\ &\int_0^1{\E_{S_{n+1} \sim \cD^{n+1}, R_{n+1} \sim \cD^{n+1}}\left[|Z_{n+1}(x)| \right]}\,dx \\
        \geq \ &\int_0^1{\E_{S_{n} \sim \cD^{n}, R_{n} \sim \cD^{n}}\left[|Z_{n}(x)| \right]}\,dx \\
        =\ &\E_{S_{n} \sim \cD^{n}, R_{n} \sim \cD^{n}}\left[ \opt(S_n, R_n)\right], 
    \end{align*}
    where the first line is because $\opt(S_{n+1}, R_{n+1})$ matches the requests of $S_{n+1}$ to those of $R_{n+1}$ greedily from left to right, and for every $x \in [0,1]$, we pay it as many times as $x$ is between a matched pair, or equivalently as large as the imbalance of points from $S_{n+1}$ are to the left of $x$ compared the points of $R_{n+1}$, the second line is by Fubini-Tonelli, the third line is by (2).
\end{proof}

\begin{proof}[Proof of \cref{lem:reduction from one sided bipartite to two sided with different distributions}]
    Let $R_2 = R \setminus R_1$ and $S_2 = S \setminus S_1$.
    Note that $|R_2| = n - |R_1| = |S_1|$ and $|S_2| = n - |S_1| = |R_1|$. 
    Let $n_S = |S_1|$ and $n_R = |R_1| = n - n_S$.
    Then we have:
    \begin{align*}
        &\E_{S \sim \cD^n, R \sim \cD^n,}\left[ \osr(S, R, \pi) \right] \\
        = \ &\E_{S_1 \sim \cD^{n_S}}[ \E_{R_2 \sim \cD^{n_S}}\left[ \cA(S_1, R_2) \right] ] + \E_{R_1 \sim \cD^{n_R}}[ \E_{S_2 \sim \cD^{n_R}}\left[ \cA(R_1, S_2) \right] ] \\
        \leq \ &a \cdot \E_{S_1 \sim \cD^{n_S}}[ \E_{R_2 \sim \cD^{n_S}}\left[ \opt(S_1, R_2) \right] ] + a \cdot \E_{R_1 \sim \cD^{n_R}}[ \E_{S_2 \sim \cD^{n_R}}\left[ \opt(R_1, S_2) \right] ] \\
        = \ &a \cdot \E_{S_1 \sim \cD^{n_S}, R_2 \sim \cD^{n_S}}[ \opt(S_1, R_2) ] + a \cdot \E_{S_2 \sim \cD^{n_R}, R_1 \sim \cD^{n_R}}[ \opt(R_1, S_2) ] \\
        \leq \ &a \cdot \E_{S_1 \sim \cD^{n}, R_2 \sim \cD^{n}}[ \opt(S_1, R_2) ] + a \cdot \E_{S_2 \sim \cD^{n}, R_1 \sim \cD^{n}}[ \opt(R_1, S_2) ] \\
        = \ &2a \cdot \E_{S \sim \cD^{n}, R \sim \cD^{n}}[ \opt(S, R) ],
    \end{align*}
    where the first line is by definition of \cref{alg:reduction from bipartite two-sided arrivals to bipartite one-sided}, the second line is due to the assumption for $\cA$, the fourth line is due to \cref{lem: monotonicity of bipartite optimal solution}, and the fifth line is by renaming the sets of requests.
\end{proof}

\subsection{Proof of \cref{thm: ratio for two sided bipartite arrivals}}

The first result we use from \cite{gupta2019stochasticonlinemetricmatching} concerns general metrics.
\begin{lemma}[Theorem 1.1 \cite{gupta2019stochasticonlinemetricmatching}]
\label{lem: gupta theorem for adversarial servers and iid requests general metrics}
    There exists an algorithm $\cA_1$ such that $\E_{R \sim \cD^n}\left[ \cA_1(S, R) \right] \leq O( \left( \log{\log{\log{n}}} \right)^2 ) \cdot \E_{R \sim \cD^n}\left[ \opt(S, R) \right]$ for any distribution $\cD$ and server set $S$ for general metrics.
\end{lemma}

The second result we use from \cite{gupta2019stochasticonlinemetricmatching} concerns tree metrics, which include the line.
\begin{lemma}[Theorem 1.2 \cite{gupta2019stochasticonlinemetricmatching}]
\label{lem: gupta theorem for adversarial servers and iid requests on the line}
    There exists an algorithm $\cA_2$ such that $\E_{R \sim \cD^n}\left[ \cA_2(S, R) \right] \leq O(1) \cdot \E_{R \sim \cD^n}\left[ \opt(S, R) \right]$ for any distribution $\cD$ and server set $S$ for tree metrics like the line.
\end{lemma}

\begin{proof}[Proof of \cref{thm: ratio for two sided bipartite arrivals}]
    For general metrics we use \cref{alg:reduction from bipartite two-sided arrivals to bipartite one-sided} with $\cA_1$, which has ratio $a = O(\log{\log{\log{n}}})^2$ by \cref{lem: gupta theorem for adversarial servers and iid requests general metrics}.
    For tree metrics we use \cref{alg:reduction from bipartite two-sided arrivals to bipartite one-sided} with $\cA_2$, which has ratio $a = 9$ by \cref{lem: gupta theorem for adversarial servers and iid requests on the line}.
\end{proof}
\newpage

\bibliographystyle{plainnat}
\bibliography{main.bib}
\appendix

\newpage
\section{Omitted Proofs from Section 3}
\label{sec:appendixframework}

\subsection{Upper Bound Cost of Framework}
\label{sec: upper bound cost of framework omitted proofs}

\UpperBoundFramework*

\ChargeInsideJBadEvent*
\begin{proof}
    We have:
    \begin{align*}
        \E_{I \sim \cD^n}\left[ \sum_{t \in [n]: r_t \in J} \charge{t} \middle| E_{bad} \right]
        &= \E_{I \sim \cD^n}\left[ \sum_{t \in [n]: r_t, r_{\tau(t)} \in J} \charge{t} \middle| E_{bad} \right] \\
        &\leq \E_{I \sim \cD^n}\left[ \sum_{t \in [n]: r_t, r_{\tau(t)} \in J} \frac{\sup{J} - \inf{J}}{2} \middle| E_{bad} \right] \\
        &= \frac{n}{2} \cdot (\sup{J} - \inf{J}),
    \end{align*}
    where the first line is because if $\tau(t) \notin J$, then $\charge{t} = 0$ by definition, the second line because if both $r_t, r_{\tau(t)} \in J$, then $\charge{t} = \frac{|r_t - r_{\tau(t)}|}{2} \leq \frac{\sup{J} - \inf{J}}{2}$ by definition, and the third line is by summing $n$ times. The statement follows since $n/2 \leq n$.
\end{proof}

\CostGoodEvent*
\begin{proof}
    We have:
    \begin{align*}
        \E&_{I \sim \cD^n}\left[ \sum_{t \in [n]} \cost{t} \middle| \bar{E}_{bad} \right] = \\
        &= \E_{I \sim \cD^n}\left[ \sum_{t \in [t_{\rho}]} \cost{t} +  \sum_{t \in [n] \setminus [t_{\rho}]} \cost{t} \middle| \bar{E}_{bad} \right] \\
        &= \E_{I \sim \cD^n}\left[ \sum_{t \in [t_{\rho}]: i(t) \in \{\leftBucket,\rightBucket\}} \cost{t} + \sum_{t \in [t_{\rho}]: i(t) \notin \{\leftBucket,\rightBucket\}} \cost{t} + \sum_{t \in [n] \setminus [t_{\rho}]} \cost{t} \middle| \bar{E}_{bad} \right]. \tag{2}
    \end{align*}
    For the first term we have:
    \begin{align*}
        \E&_{I \sim \cD^n}\left[ \sum_{t \in [t_{\rho}]: i(t) \in \{\leftBucket,\rightBucket\}} \cost{t} \middle| \bar{E}_{bad} \right] \leq \\
        &\leq \sum_{k=1}^{\rho} \E\left[ \sum_{j = 1}^{\beta^{\rho + 1 - k}} | \{  t \in \{t_{k-1} + 1, \dots, t_k \}: i(t) = j\} | \cdot ( \sup{I^k_j} - \inf{I^k_j} ) \middle| \bar{E}_{bad} \right] \\
        &\leq \E\left[ \sum_{j = 1}^{\beta^{\rho}} | \{  t \in \{t_{\learnBuckets} + 1, \dots, t_1 \}: i(t) = j\} | \cdot ( \sup{I^1_j} - \inf{I^1_j} ) \middle| \bar{E}_{bad} \right] \\
        & \qquad +  \sum_{k=2}^{\rho} \E\left[ \sum_{j = 1}^{\beta^{\rho + 1 - k}} | \{  t \in \{t_{k-1} + 1, \dots, t_k \}: i(t) = j\} | \cdot ( \sup{I^k_j} - \inf{I^k_j} ) \middle| \bar{E}_{bad} \right] \\
        &\leq \E\left[32 \beta^2 \cdot ( \sup{I^1_j} - \inf{I^1_j} ) \middle| \bar{E}_{bad} \right] + \sum_{k=2}^{\rho} \E\left[ \sum_{j = 1}^{\beta^{\rho + 1 - k}} 24 \beta \cdot ( \sup{I^k_j} - \inf{I^k_j} ) \middle| \bar{E}_{bad} \right] \\
        &= 32 \beta^2 \cdot \E\left[ \sup{I^1_j} - \inf{I^1_j}  \middle| \bar{E}_{bad} \right] + \sum_{p=1}^{\rho}  12 \beta \cdot \E\left[ \sum_{j = 1}^{\beta^{\rho + 1 - k}} ( \sup{I^k_j} - \inf{I^k_j} ) \middle| \bar{E}_{bad} \right] \\
        &\leq 96\beta^2 \cdot \E\left[ \max_{t \in [n]}r_{t} - \min_{t \in [n]}r_{t} \middle| \bar{E}_{bad} \right] + \sum_{k=2}^{\rho}  36 \beta \cdot \E\left[ \max_{t \in [n]}r_{t} - \min_{t \in [n]}r_{t} \middle| \bar{E}_{bad} \right] \\
        &= (96\beta^2 +36 \rho \beta) \cdot \E\left[ \max_{t \in [n]}r_{t} - \min_{t \in [n]}r_{t} \middle| \bar{E}_{bad} \right],  \tag{2a}
    \end{align*}
    where the first line is due to $\bar{E}_{bad} \subseteq \bar{A}_{\rho}$, which implies our algorithm is not forced to make arbitrary assignments in the first $\rho$ phases, and each request that arrives in bucket $B_j$ pays at most the length of $I^k_j$ as matching cost. The second line is due to splitting the sum into the cost of the first phase and the cost of the subequent phases. The third line is due  $\bar{E}$ implies $\bar{E}_a$, i.e. no bucket in the first phase receives more than $32\beta^2$ requests, and no bucket in a later phase receives more than $24\beta$ requests. The third line is by linearity of expectation, the fourth line is due to \cref{cond: buckets have bounded length}. 

    For the second term we have:
    \begin{align*}
        \E_{I \sim \cD^n}\left[ \sum_{t \in [t_{\rho}]: i(t) \in \{L, R\}} \cost{t} \middle| \bar{E}_{bad} \right]
        &\leq \E_{I \sim \cD^n}\left[ \sum_{t \in [t_{\rho}]: i(t) \in \{L, R\}} \left( \max_{t' \in [n]}r_{t'} - \min_{t' \in [n]}r_{t'} \right) \middle| \bar{E}_{bad} \right] \\
        &\leq \beta^2 \cdot \E_{I \sim \cD^n}\left[ \max_{t \in [n]}r_{t} - \min_{t \in [n]}r_{t} \middle| \bar{E}_{bad} \right], \tag{2b}
    \end{align*}
    where the first line is because $\cost{t} \leq |r_t - r_{\tau{t}}| \leq \max_{t' \in [n]}r_{t'} - \min_{t' \in [n]}r_{t'}$, the second line is due to that fact $\bar{E}_{bad} \subseteq \bar{E}_o$, i.e., at most $\beta^2$ requests arrive outside of the buckets.

    Similarly, for the third term we have: 
    \begin{align*}
        \E_{I \sim \cD^n}\left[ \sum_{t \in [n] \setminus [t_{\rho}]} \cost{t} \middle| \bar{E}_{bad} \right]
        &\leq \E_{I \sim \cD^n}\left[ \sum_{t \in [n] \setminus [t_{\rho}]} \left( \max_{t' \in [n]}r_{t'} - \min_{t' \in [n]}r_{t'} \right) \middle| \bar{E}_{bad} \right] \\
        &= \left( n - t_{\rho} \right) \cdot \E_{I \sim \cD^n}\left[ \max_{t \in [n]}r_{t} - \min_{t \in [n]}r_{t} \middle| \bar{E}_{bad} \right] \\
        &= 3 \beta \cdot  \E_{I \sim \cD^n}\left[ \max_{t \in [n]}r_{t} - \min_{t \in [n]}r_{t} \middle| \bar{E}_{bad} \right], \tag{2c}
    \end{align*}
    where $n-t_{\rho} = 3\beta$ follows from the definition of $t_{\rho}$.
    Also we have:
    \begin{align*}
        1 - \frac{\rho + 4}{n} 
        \geq 1 - \frac{\log n + 4}{n}  
        \geq \frac{1}{4}, \tag{2d}
    \end{align*}
    where the first inequality is because $\rho \leq \frac{\log{\frac{n}{\beta^\rho}}}{\log{\beta}} \leq \log{n}$, and the second inequality is by assumption $n \geq 2 \cdot 10^6$.

    By substituting Equations (2a), (2b), (2c) to (2) we have:
    \begin{align*}
        \E_{I \sim \cD^n}\left[ \alg(I) \middle| \bar{E}_{bad} \right]
        &\leq \left(96\beta^2 + 36\beta\rho + 3\beta + \beta^2 \right) \cdot \E_{I \sim \cD^n}\left[ \max_{t \in [n]}r_{t} - \min_{t \in [n]}r_{t} \middle| \bar{E}_{bad} \right] \\
        &\leq \left(96\beta^2 +36\beta\rho +  3\beta + \beta^2 \right) \cdot \frac{\E_{I \sim \cD^n}\left[ \max_{t \in [n]}r_{t} - \min_{t \in [n]}r_{t} \right]}{\Pr\left[ \bar{E}_{bad} \right]} \\
        &\leq \left(96\beta^2 +36\beta\rho +  3\beta + \beta^2 \right) \cdot \frac{\E_{I \sim \cD^n}\left[ \max_{t \in [n]}r_{t} - \min_{t \in [n]}r_{t} \right]}{1 - \frac{\rho + 4}{n}} \\
        &\leq \left(288\beta^2 + 144\beta\rho +  12\beta + 4\beta^2 \right) \cdot \E_{I \sim \cD^n}\left[ \max_{t \in [n]}r_{t} - \min_{t \in [n]}r_{t} \right], \tag{3}
    \end{align*}
    where the second equality is by tower rule, the third line is due to $\Pr\left[ \bar{E}_{bad} \right] \geq 1 - \frac{6\rho\beta + 2}{n}$ by \cref{cond: probability of hitting wall,cond: few arrivals in any one bucket,cond: buckets have bounded length} and union bound, and the fourth line is due to Equation (2d).
    Then the statement follows from $\beta, \rho = O(\log{n})$.
\end{proof}

\ChargeOutsideJ*
\begin{proof}
    We have:
    \begin{align*}
        \E&_{I \sim \cD^n}\left[ \sum_{t \in [n]: r_t \notin J}\charge{t} \right] = \\
        &= \E_{I \sim \cD^n}\left[ \sum_{t \in [n]: r_t < \inf{J}}\charge{t} \right] + \E_{I \sim \cD^n}\left[ \sum_{t \in [n]: r_t > \sup{J}}\charge{t} \right] \\
        &= \E_{I \sim \cD^n}\left[ \sum_{t \in [n]: r_t < \inf{J}} (\sup{J} - r_t) \right] + \E_{I \sim \cD^n}\left[ \sum_{t \in [n]: r_t > \sup{J}}( r_t - \inf{J}) \right], \tag{1}
    \end{align*}
    where the first line is by linearity of expectation, and the second line is by definition of $\charge{t}$.
    For the first term we have:
    \begin{align*}
        \E&_{I \sim \cD^n}\left[ \sum_{t \in [n]: r_t < \inf{J}} (\sup{J} - r_t) \right] = \\
        &= \E_{I \sim \cD^n}\left[ \sum_{t \in [n]: r_t < \inf{J}} (\inf{J} - r_t + \sup{J} - \inf{J} ) \right] \\
        &= \E_{I \sim \cD^n}\left[ \sum_{t \in [n]: r_t < \inf{J}} (\inf{J} - r_t) \right] + \E_{I \sim \cD^n}\left[ \sum_{t \in [n]}\mathbbm{1}\{r_t < \inf{J}\} \right] \cdot (\sup{J} - \inf{J}) \\
        &\leq \E_{I \sim \cD^n}\left[ \sum_{t \in [n]} (\inf{J} - r_t)^+ \right] + \E_{I \sim \cD^n}\left[ \sum_{t \in [n]} \frac{1}{n} \right] \cdot (\sup{J} - \inf{J}) \\
        &= n \cdot \E_{r \sim \cD}\left[ (\inf{J} - r)^+ \right] + (\sup{J} - \inf{J}), \tag{1a}
    \end{align*}
    where the first line is by adding and subtracting $\inf{J}$, the second line is due to linearity of expectation, the third line is because $\Pr_{r \sim \cD}\left[ r < \inf{J} \right] \leq 1/n$ by definition, and the fourth line is because we sum $n$ times.
    For the second term we have:
    \begin{align*}
        \E&_{I \sim \cD^n}\left[ \sum_{t \in [n]: r_t > \sup{J}}( r_t - \inf{J}) \right] = \\
        &= \E_{I \sim \cD^n}\left[ \sum_{t \in [n]: r_t > \sup{J}} (r_t - \sup{J} + \sup{J} - \inf{J} ) \right] \\
        &= \E_{I \sim \cD^n}\left[ \sum_{t \in [n]: r_t > \sup{J}} (r_t - \sup{J}) \right] + \E_{I \sim \cD^n}\left[ \sum_{t \in [n]}\mathbbm{1}\{r_t > \sup{J}\} \right] \cdot (\sup{J} - \inf{J}) \\
        &\leq \E_{I \sim \cD^n}\left[ \sum_{t \in [n]} (r_t - \sup{J})^+ \right] + \E_{I \sim \cD^n}\left[ \sum_{t \in [n]} \frac{1}{n} \right] \cdot (\sup{J} - \inf{J}) \\
        &= n \cdot \E_{r \sim \cD}\left[ (r - \sup{J})^+ \right] + (\sup{J} - \inf{J}), \tag{1b}
    \end{align*}
    where the first line is by adding and subtracting $\sup{J}$, the second line is due to linearity of expectation, the third line is because $\Pr_{r \sim \cD}\left[ r > \sup{J} \right] \leq 1/n$ by definition, and the fourth line is because we sum $n$ times. 
    We obtain the statement by substituting Equations (1a), (1b) to (1).
\end{proof}

\MaxMinRegions*

To prove this lemma, we first need two claims. Define $g(p):=1-p^n-(1-p)^n$.
\begin{claim}
\label{claim:np4}
If $0\le p\le 1/n$ and $n \geq 2$, then $g(p)\ge \frac{np}{4}$.
\end{claim}
\begin{proof}
    First, we have, $(1-p)^n = e^{n \log(1-p)} \le e^{-np} \leq 1 - np / 2$ where the last inequality is since $np \leq 1$.
    Moreover,  $p^n \le p (1/n)^{n-1} \leq  \frac{p}{n} \le \frac{np}{4}$ where the first inequality is since  $0\le p\le 1/n$, the second is since $n\ge 2$, and the third since $n^2 \geq 4$. Combining these bounds gives $g(p) \geq 1 - np/4 - (1 - np/2) \ge np/4$.
\end{proof}

\begin{claim}
\label{claim:12e}
    If $p\in[1/n,\,1-1/n]$, then $g(p)\ge 1-2/e$.
\end{claim}
\begin{proof}
    We have 
    $p^n\le(1-1/n)^n \leq 1/e$ and $(1-p)^n\le(1-1/n)^n \leq 1/e$, so
    $g(p)\ge 1-2/e$.
\end{proof}

We denote by $F(x)$ the CDF of distribution $\cD$.
\begin{proof}[Proof of \cref{lem:maxminregions}]
    Note that
    \begin{align*}
    \mathbb E[\max_{t \in [n]} r_t-\min_{t \in [n]} r_t]
     & =\int_0^1\Pr[\max_{t \in [n]} r_t>x]\,dx-\int_0^1\Pr[\min_{t \in [n]} r_t>x]\,dx \\
     & =\int_0^1(1 - F(x)^n)\,dx-\int_0^1(1-F(x))^n\,dx \\
     & =\int_0^1 \Big(1-F(x)^n-(1-F(x))^n\Big)\,dx.
\end{align*}

We now split the integral into three regions.

\smallskip
\noindent\emph{Left region $[0,\inf{J}]$.}
For $x<\inf{J}$, we have $F(x)\le 1/n$, so by Claim~\ref{claim:np4},
\[
    g(F(x))\ge \frac{nF(x)}{4}.
\]
Hence
\[
    \int_0^{\inf{J}} g(F(x))\,dx
    \ge \frac{n}{4}\int_0^{\inf{J}}F(x)\,dx
    = \frac{n}{4} \int_0^{\inf{J}}\Pr_{r \sim \cD}[r\le x]\,dx
    = \frac{n}{4} \E_{r \sim \cD}\left[ (\inf{J} - r)^+ \right].
\]

\smallskip
\noindent\emph{Middle region $[\inf{J},\sup{J}]$.}
For $x\in(\inf{J},\sup{J})$, we have $F(x)\in[1/n,\,1-1/n]$, so by Claim~\ref{claim:12e}, $g(F(x))\ge 1-\frac{2}{e}.$ Therefore
\[
\int_{\inf{J}}^{\sup{J}} g(F(x))\,dx
\ge \Bigl(1-\frac{2}{e}\Bigr)(\sup{J}-\inf{J}).
\]

\smallskip
\noindent\emph{Right region $[\sup{J},1]$.}
For $x>\sup{J}$, we have 
\[
    g(F(x)) = g(1 - F(x)) \ge \frac{n(1-F(x))}{4}
    = \frac{n}{4}\Pr[r>x].
\]
where the first equality is since $g(p) = g(1-p)$ for all $p$ and the inequality by Claim~\ref{claim:np4} and since $1-F(x)\le 1/n$. Thus
\[
    \int_{\sup{J}}^1 g(F(x))\,dx
    \ge \frac{n}{4}\int_{\sup{J}}^1\Pr[r>x]\,dx
    = \frac{n}{4} \E_{r \sim \cD}\left[ (r - \sup{J})^+ \right].
\]
By combining the three regions, we get the desired inequality.
\end{proof}

\begin{proof}[Proof of \cref{lem: upper bound cost of framework}]
    For $n \leq 2 \cdot 10^6$:
    \begin{align*}
        \E_{I \sim \cD^n}\left[ \alg(I) \right] 
        = \E_{I \sim \cD^n}\left[ \sum_{t \in [n]}\cost{t} \right]
        \leq n \cdot \E_{I \sim \cD^n}\left[ \max_{t \in [n]}r_t - \min_{t \in [n]}r_t \right]
        \leq 2\cdot10^6 \cdot \E_{I \sim \cD^n}\left[ \opt(I) \right],
    \end{align*}
    where the equality is by definition of $\cost{t}$, the first inequality is because $\cost{t} \leq \max_{t\in[n]}r_t - \min{t \in [n]}r_t$ for any $t \in [n]$ and the second inequality is by assumption on $n \geq 2 \cdot 10^6$ and \cref{lemma: opt pays expected rightmost minus expected leftmost}.
    The remaining proof assumes $n \geq 2 \cdot 10^6$:
    \begin{align*}
        \E_{I \sim \cD^n}\left[ \alg(I) \right] 
        = \E_{I \sim \cD^n}\left[ \sum_{t \in [n]}\cost{t} \right]
        \leq \E_{I \sim \cD^n}\left[ \sum_{t \in [n]}\charge{t} \right], \tag{1}
    \end{align*}
    where the equality is because we sum the cost paid in each time step $t$ over all $t \in [n]$ and the inequality is because for all $t \in [n]$, we have $\cost{t} + \cost{\tau(t)} \leq \charge{t} + \charge{\tau(t)}$ by definition.
    By linearity of expectation:
    \begin{align*}
        \E_{I \sim \cD^n}\left[ \sum_{t \in [n]}\charge{t} \right] 
        = \E_{I \sim \cD^n}\left[ \sum_{t \in [n]: r_t \in J}\charge{t} \right] + \E_{I \sim \cD^n}\left[ \sum_{t \in [n]: r_t \notin J}\charge{t} \right]. \tag{2}
    \end{align*}
    For the first term we have:
    \begin{align*}
        \E&_{I \sim \cD^n}\left[ \sum_{t \in [n]: r_t \in J}\charge{t} \right] \leq \\
        &\leq \E_{I \sim \cD^n}\left[ \sum_{t \in [n]: r_t \in J}\charge{t} \middle| E_{bad} \right] \cdot  \frac{6\rho \beta + 2}{n} + \E_{I \sim \cD^n}\left[ \sum_{t \in [n]: r_t \in J}\charge{t} \middle| \bar{E}_{bad} \right] \\
        &\leq \E_{I \sim \cD^n}\left[ \sum_{t \in [n]: r_t \in J}\charge{t} \middle| E_{bad} \right] \cdot  \frac{6\rho \beta + 2}{n} + \E_{I \sim \cD^n}\left[ \sum_{t \in [n]: r_t \in J}\cost{t} \middle| \bar{E}_{bad} \right] \\
        &\leq O(\log^2{n}) \cdot \left( \sup{J} - \inf{J} \right) + O\left( \log^2{n} \right) \cdot \E_{I \sim \cD^n}\left[ \max_{t \in [n]}r_t - \min_{t \in [n]}r_t \right] \\
        &\leq O\left( \log^2{n} \right) \cdot \E_{I \sim \cD^n}\left[ \max_{t \in [n]}r_t - \min_{t \in [n]}r_t \right], \tag{2a}
    \end{align*}.
    where the first line is due to $\Pr_{I \sim \cD^n}\left[ E_{bad} \right] \leq \Pr_{I \sim \cD^n}\left[ \bar{A}_{\rho} \right] + \Pr_{I \sim \cD^n}\left[ \bar{E}_{a} \right] + \Pr_{I \sim \cD^n}\left[ \bar{E}_{o} \right] \leq \frac{6\rho \beta + 2}{n}$ by \cref{cond: probability of hitting wall,cond: few arrivals in any one bucket} and union bound, the second line is because for all $t$ such that $r_t \in J$ we have $\charge{t} \leq \cost{t}$ by definition, the third line is due to \cref{lem: charge requests in J for bad event,lem: total cost for good event} which hold by assumption $n \geq 2 \cdot 10^6$ and $\rho, \beta = O(\log{n})$, and the fourth line is due to \cref{lem:maxminregions}.
    For the second term we have:
    \begin{align*}
        \E_{I \sim \cD^n}\left[ \sum_{t \in [n]: r_t \notin J}\charge{t} \right]
        &\leq 2(\sup{J} - \inf{J}) + n \cdot \E_{r \sim \cD}\left[\inf_{x \in J} |r - x|\right] \\
        &\leq 8 \cdot \E_{I \sim \cD^n}\left[ \max_{t \in [n]}r_t - \min_{t \in [n]}r_t \right], \tag{2b}
    \end{align*}
    where the first line is due to \cref{lem: charge requests outside J}, and the second is due to \cref{lem:maxminregions}.
    Then we obtain the result by substituting Equations (2a), (2b) to (2) and chaining with (1).
\end{proof}

\subsection{Lower Bound OPT}
\label{sec: opt lower bound omitted proofs}

\OPTMaxMin*
\begin{proof}
    Let random variables $r^{(1)}, \dots, r^{(n)}$ such that $r^{(i)}$ corresponds to the $i$-th leftmost request and $\cD^{(i)}$ be its distribution.
    We condition on $r^{(1)}, r^{(n)}$ and show that for any point $x \in \left( r^{(1)}, r^{(n)} \right)$, the probability that an odd number of requests arrive between $r^{(1)}$ and $x$ is at least $\frac{1}{2}$. This implies that with probability $\frac{1}{2}$, $x$ can be charged to $\opt$. 
    We calculate the expectation of the optimal solution as follows:
    \begin{align*}
        \E_{I \sim \cD^n}[\opt(I)] 
        = \E_{r' \sim \cD^{(1)}, r'' \sim \cD^{(n)}}\left[ \E_{I \sim \cD^n | r^{(1)}, r^{(n)}}\left[\opt(I) \mid r^{(1)} = r', r^{(n)} = r'' \right] \right]. \tag{*}
    \end{align*}
    We define the number of arrivals left of $x$ for all $x \in (r^{(1)}, r^{(n)})$ as $Y_x$:
    \begin{align*}
        Y_x &= | \{r_1, \dots, r_n\} \cap (r^{(1)}, x) | 
        = \sum_{i \in [n]} \mathbbm{1}\{ r_i \in (r^{(1)}, x) \} 
        = \sum_{i \in [n] \setminus \{\arg\min_{t \in [n]}r_t, \arg\max_{t \in [n]}r_t\}} \mathbbm{1}\{ r_i \in (r^{(1)}, x) \}.
    \end{align*}  
    Notice that for $p_x(r') = \Pr_{r \sim \cD}\left[r \in (r', x) \right]$:
    \begin{align*}
        \Pr_{I \sim \cD^n | r^{(1)}, r^{(n)}}\left[ Y_x = k \middle| r^{(1)} = r', r^{(n)} = r'' \right] 
        = \Pr\left[ Bin(n-2, p_x(r')) = k \right]
    \end{align*}
    because for $i,j \in [n] \setminus \{\arg\min_{t \in [n]}r_t, \arg\max_{t \in [n]}r_t\}$ and conditional on $r^{(1)}, r^{(n)}$, we know that indicators $\mathbbm{1}\{ r_i \in (r^{(1)}, x) \}, \mathbbm{1}\{ r_j \in (r^{(1)}, x) \}$ are independent.
    Then we have that:
    \allowdisplaybreaks
    \begin{align}
    \E&_{I \sim \cD^n | r^{(1)}, r^{(n)}} \left[\opt(I) \mid r^{(1)} = r', r^{(n)} = r'' \right]
    \nonumber \\
    &= \E_{I \sim \cD^n | r^{(1)}, r^{(n)}}\left[ \sum_{i=1}^{n/2} \left( r^{(2i)} - r^{(2i-1)} \right) \middle| r^{(1)} = r', r^{(n)} = r'' \right] \\
    &= \E_{I \sim \cD^n | r^{(1)}, r^{(n)}}\left[ \int_{r'}^{r''}{\mathbbm{1}\{ | \{r_1, \dots, r_n\} \cap [0, x) | \in \{1, 3, \dots, n-1\} \}} \ dx \middle| r^{(1)} = r', r^{(n)} = r'' \right] \\
    &= \E_{I \sim \cD^n | r^{(1)}, r^{(n)}}\left[ \int_{r'}^{r''}{\mathbbm{1}\{ | \{r_1, \dots, r_n\} \cap (r^{(1)}, x) | \in \{0, 2, \dots, n-2\} \}} \ dx \middle| r^{(1)} = r', r^{(n)} = r'' \right] \\
    &= \int_{r'}^{r''}{\E_{I \sim \cD^n | r^{(1)}, r^{(n)}}\left[ \mathbbm{1}\{ Y_x \in \{0, 2, \dots, n-2\} \}  \middle| r^{(1)} = r', r^{(n)} = r'' \right]} \ dx \\
    &= \int_{r'}^{r''}{\Pr_{I \sim \cD^n | r^{(1)}, r^{(n)}}\left[  Y_x \in \{0, 2, \dots, n-2\}  \middle| r^{(1)} = r', r^{(n)} = r'' \right]} \ dx \\
    &= \int_{r'}^{r''}{\Pr\left[  Bin(n-2, p_x(r')  \in \{0, 2, \dots, n-2\} \right]} \ dx \\
    &= \int_{r'}^{r''}{ \left( \frac{1}{2} + \frac{(2p_x(r') - 1)^{n-2}}{2} \right) } \ dx \\
    &= \int_{r'}^{r''}{ \left( \frac{1}{2} + \frac{\left( (2p_x(r') - 1)^{\frac{n-2}{2}} \right)^2}{2} \right) } \ dx \\
    &\geq \int_{r'}^{r''} \frac{1}{2} \;dx \\
    &= \frac{1}{2} \cdot \left(r''-r'\right),  \tag{**}
    \end{align} 
    where (1) is due to \cref{lemma: OPT structure}, (2) is because the sum is only paying for the segments between $r^{(2i-1)}$ and $r^{(2i)}$ so only for $x \in (r^{(1)}, r^{(n)})$ such that there is an odd number of requests that arrived in $[0, x]$, (3) is by excluding $r^{(1)}$, (4) is by definition of $Y_x$, (5) is due to the indicator, (6) is because the conditional distribution of $Y_x$ is $Bin(n-2, p_x(r'))$, (7) is due to \cref{lem:probability that binomial is even} since $n$ is even, (8) is because $n-2$ is even, (9) is because $\left( (2p_x(r') - 1)^{\frac{n-2}{2}} \right)^2 \geq 0$ and (**) is by integration.
    Substituting (**) in (*) we have:
    \begin{align*}
        \E_{I \sim \cD^n}[\opt(I)] 
        \geq \E_{r' \sim \cD^{(1)}, r'' \sim \cD^{(n)}}\left[ \frac{1}{2} \cdot \left( r'' - r' \right) \right]
        = \frac{1}{2} \cdot \E_{I \sim \cD^n}\left[ \max_{t \in [n]}r_t - \min_{t \in [n]}r_t \right].
    \end{align*}
\end{proof}

\OPTStructure*
\begin{proof}
    Let the cost of a matching $M$ be $\cost{M} = \sum_{\{r, r'\} \in M} |r - r'|$.
    We show that by iteratively exchanging matched pairs of a matching $M$, we can construct matching $\hat{M}^*$:
    \begin{align*}
        \hat{M}^* = \{ \{ r^{(2i)}, r^{(2i-1)} \}: i \in [n/2] \},
    \end{align*}
    with cost at most that that of $M$. Notice that $\cost{\hat{M}^*} = \sum_{i=1}^{n/2}\left( r^{(2i)} - r^{(2i - 1)} \right)$.
    Now we show the exchange. Let $k = \min\{ i \in [n/2]: \{ r^{(2i)}, r^{(2i - 1)} \} \notin M \}$ and $a,b$ such that $\{r^{(2k)}, r^{(a)}\}, \{ r^{(2k-1)}, r^{(b)} \} \in M$, then define:
    \begin{align*}
        M_1 = M \setminus \{ \{r^{(2k)}, r^{(a)}\}, \{ r^{(2k-1)}, r^{(b)} \} \} \cup \{ \{ r^{(2k)}, r^{(2k-1)} \}, \{ r^{(a)}, r^{(b)} \} \}.
    \end{align*}
    Then we have:
    \begin{align*}
        \cost{M} - \cost{M_1} 
        &= \left( |r^{(2k)} - r^{(a)}| + |r^{(2k-1)} - r^{(b)}| \right) - \left( |r^{(2k)} - r^{(2k-1)}| + |r^{(a)} - r^{(b)}| \right) \\
        &= \left( r^{(a)} - r^{(2k)} + r^{(b)} - r^{(2k-1)} \right) - \left( r^{(2k)} - r^{(2k-1)} + |r^{(a)} - r^{(b)}| \right) \\
        &= r^{(a)} + r^{(b)} - |r^{(a)} - r^{(b)}| - 2r^{(2k)} \\
        &= 2 \left( \min\left( r^{(a)}, r^{(b)} \right) - r^{(2k)} \right) \\
        &\geq 0,
    \end{align*}
    where the first line is because $M_1$ differs from $M$ in exactly two pairs, the second line is because by definition of $k$ we have $a, b > 2k$ so $r^{(a)}, r^{(b)} \geq r^{(2k)} \geq r^{(2k-1)}$, the third line is by cancelling $r^{(2k-1)}$, the fourth line is because $r^{(a)} + r^{(b)} - |r^{(a)} - r^{(b)}| = \min\left( r^{(a)}, r^{(b)} \right)$ and the fifth line is because $r^{(a)}, r^{(b)} \geq r^{(2k)}$. 
    By repeating the exchange we can obtain $M_2, M_3, \dots, \hat{M}^*$ so:
    \begin{align*}
        \cost{M} \geq \cost{M_1} \geq \cost{M_2} \geq \dots \geq \cost{\hat{M}} = \sum_{i=1}^{n/2}\left( r^{(2i)} - r^{(2i - 1)} \right).
    \end{align*}

    Let $M^*$ be a perfect matching of $r^{(1)}, \dots, r^{(n)}$ of minimum cost. i.e. $\cost{M^*} = \opt(I)$, then for $M = M^*$ we obtain the statement.
\end{proof}

\begin{lemma}
\label{lem:probability that binomial is even}
    Let $Y \sim Bin(m, p)$ for even $m$, then $\Pr\left[ Y \in \{0, 2, \dots, m\} \right] = \frac{1}{2} + \frac{\left( 2p - 1 \right)^m}{2}$.
\end{lemma}
\begin{proof}
    Let $a, b \in \mathbb{R}$, then:
    \begin{align*}
        (a + b)^m + (a - b)^m 
        = \sum_{k=0}^{m} {m \choose k} a^k (b)^{m-k} + \sum_{k=0}^{m} {m \choose k} (-1)^{m-k} a^k (b)^{m-k}
        = 2 \sum_{k \in \{0, 2, \dots, m\}} {m \choose k} a^k (b)^{m-k}. \tag{*}
    \end{align*}

    Then we have that:
    \begin{align*}
        \Pr\left[ Y \in \{0, 2, \dots, m\} \right] 
        = \sum_{k \in \{0,2, \dots, m\}} \Pr\left[ Y = k \right]
        = \sum_{k \in \{0,2, \dots, m\}} {m \choose k} p^k (1 - p)^{m-k}
        = \frac{1}{2} + \frac{(2p - 1)^m}{2},
    \end{align*}
    where the last equality is due to (*) for $a = p, b = 1 - p$.
\end{proof}

\section{Omitted Proofs from Section 4}\label{appendix: omitted proofs unknown iid}

\begin{proof}[Proof of \Cref{lemma: negative hypergeometric tail bounds}]
For the first inequality in the lemma statement, let 
\[
\mu = \frac{\learnBuckets}{n} \cdot \frac{\beta}{8} \cdot \left(\frac{n}{2\learnBuckets} +1\right) = \frac{\beta}{16} +\frac{\beta}{8}\cdot \frac{\learnBuckets}{n} \implies \frac{\beta}{16} \leq \mu \leq \frac{\beta}{12},
\] 
which follows from $\learnBuckets \leq \frac{n}{6}$. Then, let
\begin{align*}
    \delta = \frac{\frac{\beta}{8}}{\mu}-1 = \frac{\frac{\beta}{8}-\mu}{\mu} \geq \frac{\frac{\beta}{8}-\frac{\beta}{12}}{\frac{\beta}{12}} = \frac{1}{2}.
\end{align*}
So, we have
\begin{align*}
    &\Pr\left[Y \sim \textrm{NegativeHypergeometric}\left(n, n - \learnBuckets, \frac{\beta}{8}\right) \leq \frac{n}{2\learnBuckets} \cdot \frac{\beta}{8}\right]\\
    = \ &\Pr\left[X \sim \textrm{Hypergeometric}\left(n, \learnBuckets, \frac{\beta}{8} \cdot \left(\frac{n}{2\learnBuckets}+1\right)\right) \geq \frac{\beta}{8}\right]\\
    \leq \ &\Pr\left[X \geq (1+\delta) \mu\right]\\
    \leq \ &\exp\left(-\frac{\delta^2 \mu}{2+\delta}\right) \\
    = \ &\exp\left(-\frac{\mu}{10}\right)\\
    \leq \ &\exp\left(- \frac{\beta}{160}\right)\\
    \leq \ &\frac{1}{n^3},
\end{align*}
where the second line follows from Lemma~\ref{lemma: relating negative hypergeometric to hypergeometric}, the fourth line follows from Lemma~\ref{lemma: hypergeometric tail bounds}, the sixth line follows from $\mu \geq \frac{\beta}{16}$, and the last line follows from $\beta \geq 500 \log n$.

For the second inequality in the lemma statement, let $\mu = \frac{\learnBuckets}{n} \cdot \frac{\beta}{8} \cdot \left(\frac{2n}{\learnBuckets} +1\right) > \frac{\learnBuckets}{n} \cdot \frac{\beta}{8} \cdot \frac{2n}{\learnBuckets} = \frac{\beta}{4}$. Then, let 
\begin{align*}
    \delta = 1-\frac{\frac{\beta}{8}}{\mu} = \frac{\mu - \frac{\beta}{8}}{\mu} \geq \frac{\frac{\mu}{2}}{\mu} = \frac{1}{2},
\end{align*}
where the inequality follows from $\mu \geq \frac{\beta}{4}$ which implies $\frac{\mu}{2} \geq \frac{\beta}{8}$. So, we have 
\allowdisplaybreaks
\begin{align*}
    &\Pr\left[Y \sim \textrm{NegativeHypergeometric}\left(n, n - \learnBuckets, \frac{\beta}{8}\right) \geq \frac{2n}{\learnBuckets} \cdot \frac{\beta}{8}\right]\\
    = \ &1-\Pr\left[Y \sim \textrm{NegativeHypergeometric}\left(n, n - \learnBuckets, \frac{\beta}{8}\right) \leq \frac{2n}{\learnBuckets} \cdot \frac{\beta}{8}\right]\\
    = \ &1-\Pr\left[X \sim \textrm{Hypergeometric}\left(n, \learnBuckets, \frac{\beta}{8} \cdot \left(\frac{2n}{\learnBuckets}+1\right)\right) \geq \frac{\beta}{8}\right]\\
    = \ &\Pr\left[X \sim \textrm{Hypergeometric}\left(n, \learnBuckets, \frac{\beta}{8} \cdot \left(\frac{2n}{\learnBuckets}+1\right)\right) \leq \frac{\beta}{8}\right]\\
    = \ &\Pr\left[X \leq (1-\delta) \mu\right]\\
    \leq \ &\exp\left(-\frac{\delta^2 \mu}{2}\right) \\
    \leq \ &\exp\left(-\frac{\mu}{8}\right)\\
    = \ &\exp\left(-\frac{\beta}{32}\right)\\
    \leq \ &\frac{1}{n^3},
\end{align*}
where the first line follows the law of total probability, the second line from Lemma~\ref{lemma: relating negative hypergeometric to hypergeometric}, the third line from the law of total probability, the fourth line follows from Lemma~\ref{lemma: hypergeometric tail bounds}, the fifth line follows from $\delta \geq \frac{1}{2}$, the sixth line follows from $\mu \geq \frac{\beta}{4}$, and the last line follows from $\beta \geq 500 \log n$.
\end{proof}

\begin{proof}[Proof of \Cref{lemma: hypergeometric fact}]
    Let $X(I) = \sum_{t = 1}^{k} \mathbbm{1}\{ r_t \leq r_{(m:n)} \}$ be the number of requests among the first $k$ that arrive before the $m$-th largest request. 
    We now show that the distribution of $X(I)$ is $Hypergeometric(n, m, k)$. An equivalent way to interpret $X$, is that we sample $k$ indices (the indices correspond to the order statistics) out of $n$ without replacement and we count how many of the first $m$ indices we get, so:
    \begin{align*}
        \Pr_{I \sim \cD^n}\left[ \min_{t \in [k]}r_t > r_{(m:n)} \right] 
        = \Pr_{I \sim \cD^n}\left[ X(I) = 0  \right]
        = \frac{{n - m \choose k}}{{n \choose k}}
        = \prod_{i = 0}^{k-1} \frac{n - i - m}{n - i}
        \leq \left( 1 - \frac{m}{n} \right)^{k}
        \leq e^{- \frac{km}{n}}.
    \end{align*}
    For the second part of the statement, let $Y(I) = \sum_{t = 1}^{k} \mathbbm{1}\{ r_t \geq r_{(n - m:n)} \}$ and similarly $Y(I)$ is $Hypergeometric(n, m, k)$, therefore:
    \begin{align*}
        \Pr_{I \sim \cD^n}\left[ \max_{t \in [k]}r_t < r_{(n - m:n)} \right] 
        = \Pr_{I \sim \cD^n}\left[ Y(I) = 0  \right]
        = \Pr_{I \sim \cD^n}\left[ X(I) = 0  \right]
        = \Pr_{I \sim \cD^n}\left[ \min_{t \in [k]}r_t > r_{(m:)} \right].
    \end{align*}
\end{proof}

\begin{proof}[Proof of \Cref{lemma: bound on expected number any bucket receives - random order}]
    Fix an arbitrary sequence of ordered requests $\left\{r_{(1)}, \dots, r_{(n)}\right\}$. Fix $k \in [\rho], j \in \left[\beta^{\rho+1-k}\right]$. By the tower rule, we have
\begin{align*}
    \E\left[X^k_j < \beta\right] &= \E\left[X^k_j < \beta \middle| E_M\right] \cdot \Pr\left[E_M\right] + \E\left[X^k_j < \beta \middle| \bar{E}_M\right] \cdot \Pr\left[\bar{E}_M\right]\\
    &\leq \E\left[X^k_j < \beta \middle| E_M\right] + n \cdot \Pr\left[\bar{E}_M\right]\\
    &\leq \Pr\left[X^k_j < \beta \middle| E_M\right] + \frac{1}{n},
    \end{align*}
    where the second line follows from that fact that at most $n$ requests can arrive in a bucket, and the last line follows from Lemma~\ref{lemma: enough total requests in buckets - random order}. 
    
    Next, recall that, as shown in the proof of Lemma~\ref{lemma: enough requests in each bucket - random order}, $X^k_j \sim \textrm{Hypergeometric}\left(n-\learnBuckets, \sum_{j' \in \cB_j} Y_{j'}, t_k - t_{k-1}\right)$, where $Y_{j'}$ is the number of non-learning phase requests between the endpoints of the $j'$-th bucket from the first phase and $\cB_j = \left\{j' \cdot \frac{\beta}{8} : (j-1) \cdot \frac{\beta^k}{8} \leq j' \cdot \frac{\beta}{8}  \leq j \cdot \frac{\beta^k}{8}\right\}$ is the set of indices of buckets from the first phase that are combined to form the $j$-th bucket in the $k$-th phase. . Therefore, 
    \allowdisplaybreaks
    \begin{align*}
        E\left[X^k_j \middle| E_M \right] &= E\left[X^k_j \sim \textrm{Hypergeometric}\left(n-\learnBuckets, \sum_{j' \in \cB_j } Y_{j'}, t_k - t_{k-1}\right) \middle| E_M \right]\\
        &\leq E\left[X^k_j \sim \textrm{Hypergeometric}\left(n-\learnBuckets, \frac{n \beta^k}{4\learnBuckets}, t_k - t_{k-1}\right) \right]\\
        &= \frac{t_k - t_{k-1}}{n - \learnBuckets} \cdot \frac{\beta^k n}{4\learnBuckets}\\
        &\leq \frac{t_k - t_{k-1}}{2 \learnBuckets} \cdot \beta^k\\
        &=
        \begin{aligned}
            \begin{cases}
            \frac{n-3\beta^{\rho} - \beta^{\rho+1}/8}{2 \cdot \frac{\beta^{\rho+1}}{8}} \cdot \beta^k , \quad &\textrm{ if } k=1\\
            \\
            \frac{3\beta^{\rho+1-k} (\beta-1)}{2 \cdot \frac{\beta^{\rho+1}}{8}} \cdot \beta^k, & \textrm{else} 
            \end{cases}
        \end{aligned}\\
        &<
        \begin{aligned}
            \begin{cases}
            \frac{4n}{\beta^{\rho}} - 1, \quad &\textrm{ if } k=1\\
            12\beta -1 , & \textrm{else},
            \end{cases}
        \end{aligned}
    \end{align*}
    where the second line follows from stochastic dominance since $E_M$ implies $\sum_{j' \in \cB_j} Y_{j'} \leq \frac{n \beta}{4\learnBuckets} \cdot \beta^{k-1} = \frac{n\beta^{k}}{4}$ non-learning phase requests. The the third line follows from the expectation of hypergeometric random variables, and the remaining lines follow from algebra. Therefore, we have 
    \begin{align*}
    E\left[X^k_j\right] &\leq \frac{1}{n} + \begin{aligned}
            \begin{cases}
            \frac{4n}{\beta^{\rho}} - 1, \quad &\textrm{ if } k=1\\
            12\beta -1, & \textrm{else},
            \end{cases}
        \end{aligned}
        \quad\\
        &\leq \quad \begin{aligned}
            \begin{cases}
            \frac{4n}{\beta^{\rho}} , \quad &\textrm{ if } k=1\\
            12\beta , & \textrm{else}
            \end{cases}
        \end{aligned}
    \end{align*}
    We can now use the hypergeometric tail bounds from Lemma~\ref{lemma: hypergeometric tail bounds}. Let $\delta = 1$. For $j \in \left[\beta^{\rho}\right]$ let $\mu = \E\left[X^1_j\right]$. By the above calulation, $\mu \leq \frac{4n}{\beta^{\rho}}$, and by the proof of Lemma~\ref{lemma: enough requests in each bucket - random order}, $\mu \geq \frac{5}{4}\beta$. Then, we have
    \begin{align*}
        \Pr\left[X^1_j \geq \frac{8n}{\beta^{\rho}}\right] \leq \Pr\left[X^1_j \geq (1+\delta) \mu\right] \leq \exp\left(-\frac{\delta^2 \mu}{2+\delta}\right) = \exp\left(-\frac{\mu}{3}\right) \leq \exp\left(-\frac{\frac{5}{4}\beta}{3}\right) \leq \frac{1}{n^2},
    \end{align*}
    where the last inequality follows from $\beta \geq 500\log n$. 

    Similarly, for $k \in [2..\rho]$ and $j \in \left[\beta^{\rho+1-k}\right]$, the above computations and the proof of Lemma~\ref{lemma: enough requests in each bucket - random order} prove $\frac{5}{4}\beta \leq \mu \leq 12\beta$, where $\mu = \E\left[X^k_j\right]$. Then, we have
    \begin{align*}
        \Pr\left[X^k_j \geq 24\beta\right] \leq \Pr\left[X^1_j \geq (1+\delta) \mu\right] \leq \exp\left(-\frac{\delta^2 \mu}{2+\delta}\right) = \exp\left(-\frac{\mu}{3}\right) \leq \exp\left(-\frac{\frac{5}{4}\beta}{3}\right) \leq \frac{1}{n^2}.
    \end{align*}

    Since $\frac{n}{\beta^{\rho}} \leq 4\beta^2$, we conclude
    \begin{align*}
        \Pr\left[\cup_{j \in \left[\beta^{\rho}\right]} \left\{X^1_j \geq 32\beta^2\right\}\right] \geq \Pr\left[\cup_{j \in \left[\beta^{\rho}\right]} \left\{X^1_j \geq \frac{8n}{\beta^{\rho}}\right\}\right] \leq \sum_{j \in \left[\beta^{\rho}\right]} \Pr\left[X^1_j \geq \frac{8n}{\beta^{\rho}}\right] \leq \beta^{\rho} \cdot \frac{1}{n^2} < \frac{1}{n},
    \end{align*}
    and
    \begin{align*}
        \Pr\left[\cup_{k \in [2..\rho]}\cup_{j \in \left[\beta^{\rho+1-k}\right]} \left\{X^k_j \geq 24\beta\right\}\right] \leq \sum_{k \in [2..\rho]} \sum_{j \in \left[\beta^{\rho+1-k}\right]} \Pr\left[X^k_j \geq 24\beta\right] \leq \sum_{k \in [2..\rho]} \beta^{\rho+1-k} \cdot \frac{1}{n^2} \leq \beta^{\rho} \cdot \frac{1}{n^2} < \frac{1}{n}.
    \end{align*}
\end{proof}

\end{document}